\documentclass[12pt,notitlepage]{article}%
\usepackage{amsthm}
\usepackage{setspace}
\usepackage{eurosym}
\usepackage{subfigure}
\usepackage{xcolor}
\usepackage{amssymb}
\usepackage{mathrsfs}
\usepackage{graphicx}
\usepackage[colorlinks,linkcolor=links,citecolor=cites,urlcolor=MyDarkBlue]%
{hyperref}
\usepackage{amsfonts}
\usepackage{amsmath}
\usepackage{amstext}
\usepackage{appendix}
\usepackage{appendix}
\usepackage{mathpazo}
\usepackage{tikz}
\usepackage{verbatim}

\usetikzlibrary{trees}
\usetikzlibrary{matrix}
\usetikzlibrary{decorations.pathreplacing}
\usetikzlibrary{topaths,calc}
\usepackage{natbib}%
\providecommand{\U}[1]{\protect\rule{.1in}{.1in}}
\newtheorem{theorem}{Theorem}

\newtheorem{definition}{Definition}
\newtheorem{example}{Example}

\newtheorem{lemma}{Lemma}

\newtheorem{proposition}{Proposition}

\numberwithin{equation}{section}

\definecolor{MyDarkBlue}{rgb}{0,0.08,0.45}
\definecolor{cites}{HTML}{324b13}
\definecolor{links}{HTML}{1a663b}
\definecolor{MyLightMagenta}{cmyk}{0.1,0.8,0,0.1}
\hypersetup{
colorlinks,citecolor=blue,filecolor=black,linkcolor=blue,urlcolor=blue
}
\usepackage[top=0.9in,bottom=0.9in,left=0.9in,right=0.9in]{geometry}

\usepackage{setspace}
\usepackage{footmisc}
\usepackage{lipsum}

\begin{document}

\title{A dynamic auction for multilateral collaboration}
\author{Chao Huang\thanks{Institute for Social and Economic Research, Nanjing Audit University. Email: huangchao916@163.com.}}
\date{}
\maketitle

\begin{abstract}
We study the problem of multilateral collaboration among agents with transferable utilities. Any group of agents can sign a contract consisting of a primitive contract and monetary transfers among the signatories. We propose a dynamic auction that finds a stable outcome when primitive contracts are gross complements for all participants.
\end{abstract}



\section{Introduction}\label{Sec_intro}

Equilibrium and stability are important solution concepts in various markets.
In the problem of allocating objects among agents, economists often design auctions to produce equilibria or stable outcomes, which are typically both efficient and immune to deviations. When a competitive equilibrium exists in a decentralized market, \cite{S60} demonstrated that the market does not necessarily converge to an equilibrium.\footnote{In his examples, \cite{S60} assumed the rate of change of the price of each commodity is proportional to the excess market demand for that commodity. Starting from any set of prices other than equilibrium, the prices in the examples do not converge to an equilibrium.} Consequently, designing an auction to find an equilibrium has the advantage of producing a desirable outcome in a centralized manner, avoiding the potential convergence issues that can arise in decentralized markets.

This paper designs an auction for multilateral collaboration, that is, allocating contracts rather than objects among agents. Our work closely relates to and builds upon the study of \cite{RY20} (henceforth, RY) on multilateral cooperation with transferable utilities. In their framework, any group of agents can collaborate by signing contracts. Each contract comprises a primitive contract and monetary transfers among the signing agents. RY established that a stable outcome exists when primitive contracts are gross complements for all agents. They defined gross complementarity as the monotonicity of agents' demand correspondences in the strong set order in an environment with externalities. In the absence of externalities, their definition is equivalent to supermodularity of agents' valuations and, as shown by \cite{Y23}, is the polar opposite of the gross substitutes condition of \cite{KC82}.\footnote{See Definition \ref{def_GC} and Lemma \ref{lma_GC}.} Gross complementarity captures the key feature of complementarities in real-life multilateral cooperations. RY also investigated a parallel setting with nontransferable utilities and complementary contracts. Notably, a stable outcome in this setting can be found by a novel one-sided Deferred Acceptance algorithm; however, there is no counterpart of this algorithm in the setting with transferable utilities.

In this paper, we devise a dynamic auction for identifying stable outcomes in multilateral cooperation with transferable utilities and gross-complement primitive contracts. There are various multi-item auction formats to choose from.\footnote{See, for example, \cite{M04} and \cite{K10} for other forms of multi-item auctions.} Dynamic price-adjustment procedures offer a key advantage: they do not necessitate information on agents' valuations, requiring only that agents report their demands at given prices. We assume agents' valuations take integer values. This assumption aligns with real-world practices, where agents' values are measured in currency denominations with an indivisible smallest unit (such as pennies or cents). However, unlike RY, we do not consider externalities in agents' preferences, as an agent in any round of an auction does not know what contracts other agents wish to sign at current prices. Stable outcomes are efficient in the absence of externalities but may not be efficient when externalities are present.

No incentive-compatible mechanism exists that can find a stable outcome in the market we study.\footnote{See Section \ref{Sec_strategy}.} We assume agents report truthfully in our auction. Previous dynamic auctions have been primarily monotone, with prices either ascending or descending. However, in our problem, monetary transfers in a contract sum to zero. Hence, when disagreements arise over a contract among the participants, it is natural for the auctioneer to increase some participants' payments while decreasing some other participants' payments within the same contract. In Section \ref{Sec_illu}, we briefly explain how we adjust prices in the auction to resolve disagreements in agents' collaborations.

\subsection{Related literature}

This paper studies multilateral matching in which agents have transferable utilities and any set of agents can collaborate by signing contracts. \cite{HK15} examined a multilateral matching market with continuous contracts. Multilateral matching with nontransferable utilities have been studied by \cite{RY20}, \cite{BH21}, and \cite{H23}. However, the literature has primarily focused on two-sided matching, which matches agents on one side to agents on the other side (e.g., \citealp{KB57}, \citealp{SS71}, and \citealp{KSY20,KSY24}). This problem is closely related to exchange economies with indivisible goods (e.g., \citealp{BM97}, \citealp{DKM01}, and \citealp{BK19}) and trading networks (e.g., \citealp{HKNOW13,HKNOW21} and \citealp{CEV21}). 

Walrasian-tâtonnement processes are useful for allocating indivisible items.\footnote{Walrasian tâtonnement is a series of price adjustments from initial prices in which the price of each item increases or decreases according to whether the excess demand is positive or negative.} \cite{DGS86} designed an ascending-price auction for allocating multiple commodities among agents with unit demands.\footnote{However, the existence of equilibria does not rely on the assumption of transferable utilities when agents have unit demands; see, e.g., \cite{G84}, \cite{Q84}, and \cite{KY86}.} \cite{A04} proposed an  ascending-price auction in the setting with homogeneous goods and decreasing marginal utilities. Both auctions are efficient and strategy-proof since they produce Vickrey-Clarke-Groves (VCG) outcomes. 

Based on the existence of a competitive equilibrium under gross substitutability (\citealp{KC82}), \cite{GS00} proposed an ascending-price auction that finds an equilibrium for allocating gross-substitute commodities; \cite{A06} proposed a strategy-proof ascending-price auction that yields VCG outcomes for this problem. \cite{SY06,SY09} generalized gross substitutability in an economy with two groups of commodities by allowing cross-group complements, and provided a double-track auction for their economy. Ascending-price auctions have been successfully used for the sale of radio spectrum licenses in many countries; see, for example, \cite{M00,M04}, \cite{K04}, and \cite{MMW10}.

Recent studies provided positive results on the existence of equilibria in economies without the assumption of transferable utilities; see \cite{FJJT19}, \cite{BEJK23}, \cite{RY23}, and \cite{NV24}. When utilities are transferable, the single improvement property (\citealp{GS99}) of gross substitutability is crucial for auction design. When utilities are imperfectly transferable, \cite{S22} showed that the law of aggregate demand is needed to recover the single improvement property. 

\subsection{Mitigating disagreements}\label{Sec_illu}

In each round of our auction, each agent reports all her demand sets at the current prices, and the auctioneer mitigates disagreements among agents in a manner similar to the Walrasian-tâtonnement process. At given prices that specify the payments in all primitive contracts, a primitive contract is called \textbf{partially agreeable} if it is included in a demand set of some participants but not included in any demand set of some other participants. A primitive contract is called \textbf{strongly demanded} by some agent at given prices if it is included in every demand set of this agent. If the current prices admit a competitive equilibrium, the outcome should involve all primitive contracts that are strongly demanded by some agent but should not involve any partially agreeable primitive contract. Hence, if there is a partially agreeable primitive contract $w$ that is strongly demanded by Ana but not included in any demand set of Bob, there is no competitive equilibrium at the current prices, and the auctioneer should resolve the disagreement between Ana and Bob on the primitive contract $w$. This is an elementary type of disagreement in the market, and it is natural for the auctioneer to mitigate the disagreement by increasing Ana's payment in the primitive contract $w$ by one unit and decreasing Bob's payment in $w$ by one unit. If Ana's valuation takes integer values, at the new prices the primitive contract $w$ is still in some demand set of Ana,\footnote{Since Ana's valuation takes integer values, at the original prices each of Ana's optimal choices yields a utility at least one unit higher than the utility of her any suboptimal choice. Hence, at the new prices Ana's original optimal choices yield a utility no less than the utility of her any other choices.} and Bob may demand some set containing $w$ as the cost of this set decreases by one unit. This price adjustment reduces the value of a Lyapunov function, which was used in the market with divisible goods (see, e.g., \citealp{AH71} and \citealp{V81}) and introduced by \cite{A06} to the problem with indivisibilities. The minimum of the Lyapunov function corresponds to a competitive equilibrium.\footnote{See Section \ref{Sec_CE}.}

We show that, under gross complementarity, the general type of disagreements that the auctioneer should resolve arise from a structure called \textbf{complements chain} that connects a partially agreeable primitive contract and a strongly demanded one.\footnote{See Proposition \ref{prop_chain}.} For instance, suppose at given prices there is an Ana-$w$-Bob-$w'$-Carol complements chain, where the primitive contract $w$ is not in any demand set of Ana, and Carol strongly demands $w'$. The primitive contract $w$ is a complement to $w'$ for Bob at the current prices in the sense that all his demand sets including $w'$ also include $w$. An equilibrium of the current prices must include the primitive contract $w'$ as Carol strongly demands it and must exclude $w$ as Ana does not demand it. Then we know that the current prices do not admit an equilibrium since Bob signs $w'$ only when he also signs $w$. Hence, the auctioneer should resolve the disagreement from this complements chain. In our auction, the auctioneer mitigate the disagreement by decreasing Ana's payment in $w$ by one unit, increasing Bob's payment in $w$ by one unit, decreasing Bob's payment in $w'$ by one unit, and increasing Carol's payment in $w'$ by one unit. This price adjustment reduces the value of the Lyapunov function by one unit. In each round of our auction, the auctioneer identifies at least one complements chain at the current prices and adjusts the prices of primitive contracts in the complements chain. Since the value of the Lyapunov function is reduced by at least one unit in each round of the auction, the auction converges to a competitive equilibrium in a finite number of steps. 

The remainder of this paper is organized as follows. Section \ref{Sec_M} introduces the model of multilateral cooperation with transferable utilities. Section \ref{Sec_analysis} analyzes agents' demands under gross complementarity and provides a procedure for verifying whether the current prices admit an equilibrium. We propose the dynamic auction in Section \ref{Sec_auction}. Omitted proofs are provided in the Appendix.

\section{Model\label{Sec_M}}

\subsection{Multilateral cooperation}

We follow the notations of \cite{RY20} to describe the market of multilateral cooperation with transferable utilities. There is a finite set $I$ of agents and a set $X$ of contracts. Each \textbf{contract} $x\in X$ is signed by a subset of agents $\mathrm{N}(x)\subseteq I$. We assume each contract is signed by at least two agents: $|\mathrm{N}(x)|\geq2$ for each $x\in X$.
Each contract $x=(w,\mathbf{t}^w)\in X$ consists of a \textbf{primitive contract} $w\in\Omega$ describing the contract's nonpecuniary part, and a vector of transfers $\mathbf{t}^w\in T^w\subseteq \mathbb{R}^I$ specifying the monetary transfers from (or to, if $t^w_i<0$) each participant. The set of primitive contracts $\Omega$ is finite. For each contract $x\in X$, its primitive contract is denoted as $\tau(x)$; for each set of contracts $Y\subseteq X$, the set of primitive contracts involved in $Y$ is also denoted as $\tau(Y)$: $\tau(Y)\equiv\{\tau(y)|y\in Y\}$. For each contract $x\in X$, its primitive contract $\tau(x)$ is associated with a set of agents $\mathrm{N}(\tau(x))$, which is the same as the set of agents signing $x$: $\mathrm{N}(\tau(x))\equiv \mathrm{N}(x)$.
The transfers specified in each primitive contract $w\in\Omega$ must sum to zero: $T^w=\{\mathbf{t}^w\in \mathbb{R}^{\mathrm{N}(w)}|\sum_{i\in \mathrm{N}(w)}t_i^w=0\}$. For each agent $i\in I$ and each set of contracts $Y\subseteq X$, we write $Y_i\equiv\{y\in Y|i\in \mathrm{N}(y)\}$ as the set of contracts signed by agent $i$. For each agent $i\in I$ and each set of primitive contracts $\Psi\subseteq\Omega$, we write $\Psi_i\equiv\{w\in \Psi|i\in \mathrm{N}(w)\}$ as the set of primitive contracts associated with agent $i$. For each set of contracts $Y\subseteq X$, we write $\mathrm{N}(Y)\equiv\bigcup_{y\in Y}\mathrm{N}(y)$.

Each agent $i\in I$ has a valuation $\mathrm{v}_i:2^{\Omega_i}\rightarrow \mathbb{R}$ over subsets of primitive contracts signed by her. We assume no externalities: Each agent's valuation depends only on the primitive contracts that involve this agent. Agents' preferences over sets of contracts are quasilinear in transfers; and agents do not sign a primitive contract twice. Formally, each agent $i$'s utility of $Y\subseteq X$ is given by
\begin{equation*}
\mathrm{u}_i(Y)=\left\{
\begin{aligned}
&\mathrm{v}_i(\tau(Y)\cap \Omega_i)-\sum_{(w,t^w)\in Y_i}t_i^w,\quad &\text{ if }\tau(x)\neq\tau(x') \text{ for all } x,x'\in Y_i,\\
&-\infty,\quad &\text{otherwise}.
\end{aligned}
\right.
\end{equation*}
Agent $i$'s utility function induces her choice correspondence $C_i : 2^{X_i}\rightrightarrows2^{X_i}$ defined by $\mathrm{C}_i(Y)\equiv\arg\max_S\{\mathrm{u}_i(S) \text{ s.t. } S\subseteq Y\}$. 

An \textbf{outcome} is a set of contracts $Y\subseteq X$ in which different contracts are associated with different primitive contracts: $\tau(x)\neq\tau(x')$ for all $x,x'\in Y$. An outcome $Y\subseteq X$ is \textbf{individually rational} if $Y_i\in \mathrm{C}_i(Y_i)$ for all $i\in I$. For any outcome $Y\subseteq X$, a set of contracts $Z\subseteq X\setminus Y$ \textbf{blocks} $Y$ if for all $i\in \mathrm{N}(Z)$, $Z_i\subseteq Y'$ for some $Y'\in \mathrm{C}_i(Y_i\cup Z_i)$. That is, $Z$ blocks $Y$ if any agent $i\in \mathrm{N}(Z)$ holding the contracts of $Y_i$ is willing to sign all contracts from $Z_i$ while possibly dropping some contracts from $Y_i$.

\begin{definition}\label{def_stable}
\normalfont
An outcome $Y\subseteq X$ is \textbf{stable} if it is individually rational and not blocked by any $Z\subseteq X\setminus Y$.
\end{definition}

Notice that if $Z\subseteq X\setminus Y$ blocks $Y$, we must have $\tau(Z)\cap\tau(Y)=\emptyset$. That is, a block to $Y$ does not involve any primitive contract in $Y$. To see this, suppose there is $z\in Z$ and $y\in Y$ such that $\tau(z)=\tau(y)=w$. Since each agent $i\in N(w)$ chooses $z$ in the presence of $y$, each agent $i$'s payment in $z$ is no more than her payment in $y$. Since all participants' payments sum up to zero in both $z$ and $y$, we know that $z$ and $y$ are the same contract. This contradicts $Z\subseteq X\setminus Y$, and thus we have $\tau(Z)\cap\tau(Y)=\emptyset$.

\subsection{Competitive equilibria}\label{Sec_CE}

The notion of competitive equilibrium provides an intermediate tool in analyzing the market of multilateral cooperation. The price-adjustment process in our auction is essentially a Walrasian-tâtonnement process. Let $K=\{(i,w)|i\in I\text{ and }w\in\Omega_i \}$ be the collection of all agent-primitive contract pairs $(i,w)$ in which the agent $i$ participates in the primitive contract $w$. A price vector $\mathbf{p}\in\mathbb{R}^K$ specifies the prices of each primitive contract for all participants of this primitive contract. For any price vector $\mathbf{p}\in\mathbb{R}^K$, let $\mathbf{p}_i=(p_i^w)_{w\in\Omega_i}$ denote the prices of the primitive contracts from $\Omega_i$ for agent $i$. A price vector $\mathbf{p}\in\mathbb{R}^K$ is called \textbf{balanced} if $\sum_{i\in N(w)}p^w_i=0$ for each $w\in \Omega$; that is, the prices within each primitive contract sum to zero. Let
\begin{equation*}
\mathbb{B}\equiv\{\mathbf{p}\in\mathbb{R}^K | \sum_{i\in N(w)}p^w_i=0 \text{ for each } w\in \Omega\} 
\end{equation*}
be the collection of all balanced price vectors.

Given a price vector $\textbf{p}_i\in \mathbb{R}^{\Omega_i}$ for primitive contracts that involve agent $i$, let $\mathrm{U}_i:2^{\Omega_i}\times\mathbb{R}^{\Omega_i}\rightarrow\mathbb{R}$ be agent $i$'s utility from choosing a set of primitive contracts $\Psi$ at price $\mathbf{p}_i$: $\mathrm{U}_i(\Psi,\mathbf{p}_i)\equiv \mathrm{v}_i(\Psi)-\sum_{w\in\Psi}p_i^w$; her demand correspondence $\mathrm{D}_i : \mathbb{R}^{\Omega_i}\rightrightarrows2^{\Omega_i}$ specifies her optimal choices of the primitive contracts at price $\mathbf{p}_i$: $\mathrm{D}_i(\textbf{p}_i)\equiv\arg\max_{\Psi\in\Omega_i}\{\mathrm{v}_i(\Psi)-\sum_{w\in\Psi}p_i^w\}$; and her indirect utility function $\mathrm{V}_i:\mathbb{R}^{\Omega_i}\rightarrow\mathbb{R}$ gives her maximum utility at $\mathbf{p}_i$:
$\mathrm{V}_i(\textbf{p}_i)\equiv\max_{\Psi\in\Omega_i}\{\mathrm{v}_i(\Psi)-\sum_{w\in\Psi}p_i^w\}$. 

\begin{definition}
\normalfont
A set of primitive contracts $\Phi\subseteq\Omega$ and a balanced price vector $\mathbf{p}\in \mathbb{B}$ constitute a \textbf{competitive equilibrium} $(\Phi,\mathbf{p})$ if $\Phi_i\in \mathrm{D}_i(\mathbf{p}_i)$ for each $i\in I$.
\end{definition}

A balanced price vector $\mathbf{p}\in \mathbb{B}$ is called an \textbf{equilibrium price vector} if there exists a set of primitive contracts $\Phi\subseteq\Omega$ such that $(\Phi,\mathbf{p})$ is a competitive equilibrium. If $\mathbf{p}\in \mathbb{B}$ is not an equilibrium price vector, we call $\mathbf{p}$ a \textbf{non-equilibrium price vector}. A set of primitive contracts $\Phi\subseteq\Omega$ is called \textbf{efficient} if $\Phi\in\arg\max_{\Phi\subseteq\Omega}\sum_{i\in I}\mathrm{v}_i(\Phi_i)$, that is, the aggregate valuation is maximized at $\Phi$. 

The following lemma summarizes the first theorem and a stronger second theorem of welfare economics, where the latter was provided by \citet[Lemma 6]{GS99} in an exchange economy with indivisible goods.

\begin{lemma}\label{lma_effi}
\normalfont
\begin{description}
\item[(i)] If $(\Phi,\mathbf{p})$ is a competitive equilibrium, then $\Phi$ is efficient.\footnote{In RY's framework with externalities, the set of primitive contracts of a competitive equilibrium is conditionally efficient but not necessarily efficient. See Proposition 3 and Example 3 of RY.}
\item[(ii)] If $\mathbf{p}\in \mathbb{B}$ is an equilibrium price vector, and $\Phi\subseteq\Omega$ an efficient set of primitive contracts, then $(\Phi,\mathbf{p})$ is a competitive equilibrium.
\end{description}
\end{lemma}

We define the Lyapunov function $\mathcal{L}:\mathbb{B}\rightarrow \mathbb{R}$ by
\begin{equation}
\mathcal{L}(\mathbf{p})\equiv\sum_{i\in I}\mathrm{V}_i(\mathbf{p}_i).
\end{equation}
This formula coincides with the conventional form that includes a term of the sum of all items' prices (see, e.g., formula (10) of \citealp{A06}). However, this term equals zero in our setting since the price vector is balanced. Lemma \ref{lma_effi} implies a characterization  for the existence of a competitive equilibrium. 
\begin{lemma}\label{lma_chara}
\normalfont
A balanced price vector $\mathbf{p}'\in \mathbb{B}$ is an equilibrium price vector if and only if $\mathbf{p}'\in\arg\min_{\mathbf{p}\in \mathbb{B}}\mathcal{L}(\mathbf{p})$ and $\mathcal{L}(\mathbf{p}')=\max_{\Phi\subseteq\Omega}\sum_{i\in I}\mathrm{v}_i(\Phi_i)$.
\end{lemma}
This lemma says that a balanced price vector is an equilibrium price vector if and only if it is a minimizer of the Lyapunov function and the minimum of the Lyapunov function is equal to the maximum of the aggregate valuation. In an exchange economy with indivisibilities, \citet[Proposition 1]{A06} derived the ``only if'' part of this lemma, and \citet[Lemma 1]{SY09} strengthened it into a characterization. RY (Lemma 11) also used the ``if'' part in their proof of the existence of a stable outcome.

\subsection{Gross complements}

Now we introduce three equivalent conditions. 

\begin{definition}\label{def_GC}
\normalfont
\begin{description}
  \item[(i)] Primitive contracts are \textbf{gross complements} for agent $i$ if for any price $\mathbf{p}_i\geq \mathbf{q}_i\in \mathbb{R}^{\Omega_i}$ and any $\Phi\in \mathrm{D}_i(\mathbf{p}_i)$, there exists $\Psi\in \mathrm{D}_i(\mathbf{q}_i)$ such that $\{w\in\Phi|p^w_i=q^w_i\}\subseteq\Psi$.
  \item[(ii)] Agent $i$'s valuation $\mathrm{v}_i$ is called \textbf{supermodular} if $\mathrm{v}_i(\Phi)+\mathrm{v}_i(\Psi)\leq \mathrm{v}_i(\Phi\cup \Psi)+\mathrm{v}_i(\Phi\cap \Psi)$ for any $\Phi,\Psi\subseteq \Omega_i$.
  \item[(iii)] Agent $i$'s demand correspondence $\mathrm{D}_i$ is called \textbf{antitone} if for any price $\mathbf{p}_i\geq \mathbf{q}_i\in \mathbb{R}^{\Omega_i}$, any $\Phi\in \mathrm{D}_i(\mathbf{p}_i)$, and any $\Psi\in \mathrm{D}_i(\mathbf{q}_i)$, we have $\Phi\cap\Psi\in \mathrm{D}_i(\mathbf{p}_i)$ and $\Phi\cup\Psi\in \mathrm{D}_i(\mathbf{q}_i)$.
\end{description}
\end{definition}

\begin{lemma}\label{lma_GC}
\normalfont
The following three statements are equivalent.
\begin{description}
  \item[(i)] Primitive contracts are gross complements for agent $i$.
  \item[(ii)] Agent $i$'s valuation $\mathrm{v}_i$ is supermodular.
  \item[(iii)] Agent $i$'s demand correspondence $\mathrm{D}_i$ is antitone.
\end{description}
\end{lemma}

\citet[Lemma 2]{RY20} defined the gross complements condition by isotonicity of the demand correspondence $\mathrm{D}_i$ and showed the equivalence between (ii) and (iii). The isotonicity of $\mathrm{D}_i$ is useful in analyzing agents' demands as it implies that $\mathrm{D}_i(\mathbf{p}_i)$ is a lattice (with respect to set inclusion) for each $\mathbf{p}_i\in\mathbb{R}^{\Omega_i}$. 

\citet[Theorem 1]{Y23} proposed the condition (i) and showed that it is equivalent to RY's definition. The condition (i) is the counterpart to the gross substitutes condition of \cite{KC82} and provides the following interpretation: Primitive contracts are gross complements for an agent if lowering the prices of some primitive contracts will not decrease her demand for any other primitive contracts.

Supermodularity has the following interpretation of increasing marginal returns: Agent $i$'s valuation $\mathrm{v}_i$ is supermodular if and only if for any $\Phi\subset\Psi\subseteq \Omega_i$ and any $w'\in\Omega_i\setminus\Psi$, $\mathrm{v}_i(\Phi\cup\{w'\})-\mathrm{v}_i(\Phi)\leq \mathrm{v}_i(\Psi\cup\{w'\})-\mathrm{v}_i(\Psi)$. The equivalence between (i) and (ii) is in contrast to the following relation found by \cite{GS99}: Submodularity is weaker than gross substitutability for a monotone valuation.

When primitive contracts are gross complements for all agents, RY (Proposition 3 and Theorem 2) proved the following statements.

\begin{lemma}\label{lma_RY}
\normalfont
If primitive contracts are gross complements for all agents, then
\begin{description}
  \item[(i)] the collection of efficient sets of primitive contracts is a nonempty lattice, and
  \item[(ii)] an outcome $Y\subseteq X$ is stable if and only if $\tau(Y)$ is the largest efficient set of primitive contracts and there is a competitive equilibrium $(\tau(Y),\mathbf{p})$ such that $t_i^w=p_i^w$ for each $(w,t^w)\in Y$ and $i\in I$.
\end{description}
\end{lemma}

The part (i) of this lemma is straightforward. There is clearly at least one efficient set of primitive contracts. Suppose $\Phi$ and $\Psi$ are two efficient sets of primitive contracts. For each agent $i\in I$, gross complementarity and Lemma \ref{lma_GC} indicate $\mathrm{v}_i(\Phi_i)+\mathrm{v}_i(\Psi_i)\leq \mathrm{v}_i(\Phi_i\cup \Psi_i)+\mathrm{v}_i(\Phi_i\cap \Psi_i)$. Adding up the inequalities of all agents, we have $\sum_{i\in I}\mathrm{v}_i(\Phi_i)+\sum_{i\in I}\mathrm{v}_i(\Psi_i)\leq\sum_{i\in I}\mathrm{v}_i(\Phi_i\cup \Psi_i)+\sum_{i\in I}\mathrm{v}_i(\Phi_i\cap \Psi_i)$. Hence, $\Phi\cup \Psi$ and $\Phi\cap \Psi$ are also efficient. 

Recall that stability is the solution concept we concern, while competitive equilibrium is an intermediate tool. The part (ii) of this lemma connects the two notions: Under gross complementarity, a stable outcome corresponds to a special competitive equilibrium, that is, a competitive equilibrium with the largest efficient set of primitive contracts.

\section{Demand analysis}\label{Sec_analysis}

At each round of our auction, if the current price vector is $\mathbf{p}\in \mathbb{B}$, each agent $i\in I$ is asked to report all her demand sets from $\mathrm{D}_i(\mathbf{p}_i)$. Then, the auctioneer's responsibility is to (i) verify whether $\mathbf{p}$ is an equilibrium price vector, (ii) if $\mathbf{p}$ is an equilibrium price vector, find out a competitive equilibrium with the largest efficient set of primitive contracts, which by Lemma \ref{lma_RY} corresponds to a stable outcome, and (iii) if $\mathbf{p}$ is a non-equilibrium price vector, determine how to adjust $\mathbf{p}$ so as to decrease the value of the Lyapunov function $\mathcal{L}(\mathbf{p})$. In this section, we analyze agents' demands under gross complementarity and provide answers to the questions (i) and (ii). We address the question (iii) in Section \ref{Sec_auction}. 

\subsection{Strongly demanded primitive contracts}\label{Sec_strong}
We say that a set $\Psi\subseteq\Omega_i$ of primitive contracts is agent $i$'s \textbf{demand set} at $\mathbf{p}_i\in\mathbb{R}^{\Omega_i}$ if $\Psi\in \mathrm{D}_i(\mathbf{p}_i)$, and that agent $i$ \textbf{demands} a primitive contract $w\in\Omega_i$ at $\mathbf{p}_i\in\mathbb{R}^{\Omega_i}$ if $w\in\Psi$ for some $\Psi\in \mathrm{D}_i(\mathbf{p}_i)$. For any agent $i\in I$ and any price vector $\mathbf{p}_i\in\mathbb{R}^{\Omega_i}$, define
\begin{equation*}
\overline{\mathrm{d}}^1_i(\mathbf{p}_i)\equiv\bigcup_{\Psi^j\in \mathrm{D}_i(\mathbf{p}_i)}\Psi^j \qquad\text{ and }\qquad \underline{\mathrm{d}}_i(\mathbf{p}_i)\equiv\bigcap_{\Psi^j\in \mathrm{D}_i(\mathbf{p}_i)}\Psi^j
\end{equation*} 
to be the union of all elements from $\mathrm{D}_i(\mathbf{p}_i)$ and the intersection of all elements from $\mathrm{D}_i(\mathbf{p}_i)$, respectively.\footnote{The meaning of the superscript ``1'' of $\overline{\mathrm{d}}^1_i(\mathbf{p}_i)$ will be clear in Section \ref{Sec_verify}.} If primitive contracts are gross complements for agent $i$, then $\overline{\mathrm{d}}^1_i(\mathbf{p}_i)$ is her largest demand set at $\mathbf{p}_i$, and $\underline{\mathrm{d}}_i(\mathbf{p}_i)$ her smallest demand set at $\mathbf{p}_i$. We say that agent $i$ \textbf{strongly demands} $w$ at $\mathbf{p}_i$ if $w\in\underline{\mathrm{d}}_i(\mathbf{p}_i)$, that is, the primitive contract $w$ is in every demand set of agent $i$ at $\mathbf{p}_i$. Let
\begin{equation*}
\mathrm{H}(\mathbf{p})\equiv\{w|w\in\underline{\mathrm{d}}_i(\mathbf{p}_i)\text{ for some }i\in I\}
\end{equation*}
be the collection of all strongly demanded primitive contracts at $\mathbf{p}\in \mathbb{B}$. If $\mathrm{H}(\mathbf{p})=\emptyset$, and primitive contracts are gross complements for all agents, then each agent's smallest demand set at $\mathbf{p}$ is $\emptyset$. In this case, $(\emptyset,\mathbf{p})$ is a competitive equilibrium but (probably) does not correspond to a stable outcome since $\emptyset$ is (probably) not the largest efficient set of primitive contracts.

\subsection{Partially agreeable primitive contracts}\label{Sec_partial}

For any balanced price vector $\mathbf{p}\in \mathbb{B}$, we call a primitive contract $w\in\Omega$ \textbf{partially agreeable} at $\mathbf{p}$ if it is demanded by some agent (i.e., $w\in\overline{\mathrm{d}}^1_j(\mathbf{p}_j)$ for some $j\in \mathrm{N}(w)$) but not demanded by some other agent (i.e.,  $w\notin\overline{\mathrm{d}}^1_{j'}(\mathbf{p}_{j'})$ for some $j'\in \mathrm{N}(w)$). Let
\begin{equation*}
\mathrm{G}^1(\mathbf{p})\equiv\{w|w\in\overline{\mathrm{d}}^1_j(\mathbf{p}_j)
\text{ and }w\notin\overline{\mathrm{d}}^1_{j'}(\mathbf{p}_{j'})\text{ for some } j,j'\in \mathrm{N}(w)\}
\end{equation*}
be the collection of all partially agreeable primitive contracts at $\mathbf{p}\in \mathbb{B}$. The notion of partially agreeable primitive contracts will be generalized in Section \ref{Sec_verify}.

Suppose primitive contracts are gross complements for all agents, and there are no partially agreeable primitive contracts: $\mathrm{G}^1(\mathbf{p})=\emptyset$. In this case, we know that $(\bigcup_{i\in I}\overline{\mathrm{d}}^1_i(\mathbf{p}_i),\mathbf{p})$ is a competitive equilibrium\footnote{This is because $\mathrm{G}^1(\mathbf{p})=\emptyset$ indicates $\overline{\mathrm{d}}^1_i(\mathbf{p}_i)=\Omega_i\cap(\bigcup_{i\in I}\overline{\mathrm{d}}^1_i(\mathbf{p}_i))\in \mathrm{D}_i(\mathbf{p}_i)$ for each agent $i\in I$.} in which each agent $i\in I$ signs the primitive contracts from her largest demand set $\overline{\mathrm{d}}^1_i(\mathbf{p}_i)$. Then, the part (i) of Lemma \ref{lma_effi} indicates that the set $\bigcup_{i\in I}\overline{\mathrm{d}}^1_i(\mathbf{p}_i)$ is efficient. Suppose $\bigcup_{i\in I}\overline{\mathrm{d}}^1_i(\mathbf{p}_i)$ is not the largest efficient set of primitive contracts, then the part (i) of Lemma \ref{lma_RY} implies that there is a larger efficient set of primitive contracts $\Phi\supset\bigcup_{i\in I}\overline{\mathrm{d}}^1_i(\mathbf{p}_i)$. The part (ii) of Lemma \ref{lma_effi} indicates that $(\Phi,\mathbf{p})$ is also a competitive equilibrium. This is impossible since $\Phi$ must contain a contract $w\in\Omega\setminus\bigcup_{i\in I}\overline{\mathrm{d}}^1_i(\mathbf{p}_i)$ that is not demanded by some agent $i\in N(w)$ at $\mathbf{p}_i$. Therefore, the set $\bigcup_{i\in I}\overline{\mathrm{d}}^1_i(\mathbf{p}_i)$ is the largest efficient set of primitive contracts, and thus, the part (ii) of Lemma \ref{lma_RY} indicates that the competitive equilibrium $(\bigcup_{i\in I}\overline{\mathrm{d}}^1_i(\mathbf{p}_i),\mathbf{p})$ corresponds to a stable outcome.

Notice that if $(\Phi,\mathbf{p})$ is a competitive equilibrium, a primitive contract $w\in H(\mathbf{p})$ is strongly demanded by some agent $i\in I$ at $\mathbf{p}_i$, and a primitive contract $w'\in G^1(\mathbf{p})$ is partially agreeable at $\mathbf{p}$, then $w\in\Phi$ and $w'\notin\Phi$. Hence, if at $\mathbf{p}\in \mathbb{B}$ there exists a primitive contract $w\in\mathrm{G}^1(\mathbf{p})\cap \mathrm{H}(\mathbf{p})$ that is both strongly demanded by an agent and partially agreeable as it is not demanded by another agent, then $\mathbf{p}$ is a non-equilibrium price vector. Now suppose there is no primitive contract of this sort at $\mathbf{p}$ (i.e., $\mathrm{G}^1(\mathbf{p})\cap \mathrm{H}(\mathbf{p})=\emptyset$). We remove all partially agreeable primitive contracts and let $\Omega^2(\mathbf{p})=\Omega\setminus \mathrm{G}^1(\mathbf{p})$ be the set of remaining primitive contracts. If primitive contracts are gross complements for all agents, each agent's smallest demand set at $\mathbf{p}$ belongs to $\Omega^2(\mathbf{p})$ since $\mathrm{G}^1(\mathbf{p})\cap \mathrm{H}(\mathbf{p})=\emptyset$. It means that each agent has at least one demand set that belongs to $\Omega^2$. For each agent $i\in I$, let $\overline{\mathrm{d}}^2_i(\mathbf{p}_i)\equiv\bigcup_{\Psi^s\in \mathrm{D}_i(\mathbf{p}_i):\Psi^s\subseteq\Omega^2}\Psi^s$ be the union of elements from $\mathrm{D}_i(\mathbf{p}_i)$ that belong to $\Omega^2$, which is agent $i$'s largest demand set confined in $\Omega^2(\mathbf{p})$. Let $\mathrm{G}^2(\mathbf{p})\equiv\{w|w\in\overline{\mathrm{d}}^2_j(\mathbf{p}_j)$ and $w\notin\overline{\mathrm{d}}^2_{j'}(\mathbf{p}_{j'})$ for some $j,j'\in \mathrm{N}(w)\}$ be the set of ``partially agreeable'' primitive contracts when we only consider agents' demand sets confined in $\Omega^2(\mathbf{p})$. Again, we have the following observations: (i) If $\mathrm{G}^2(\mathbf{p})=\emptyset$, then $(\bigcup_{i\in I}\overline{\mathrm{d}}^2_i(\mathbf{p}_i),\mathbf{p})$ corresponds to a stable outcome; (ii) if $\mathrm{G}^2(\mathbf{p})\cap \mathrm{H}(\mathbf{p})\neq\emptyset$, then $\mathbf{p}$ is a non-equilibrium price vector;\footnote{Notice that if $(\Phi,\mathbf{p})$ is a competitive equilibrium, then $w\in G^1(\mathbf{p})$ implies $w\notin\Phi$, and thus we know $\Phi\subseteq\Omega^2(\mathbf{p})$; and we further know that $w\in G^2(\mathbf{p})$ also implies $w\notin\Phi$; but since $w\in H(\mathbf{p})$ implies $w\in\Phi$, we know that $\mathbf{p}$ is a non-equilibrium price vector if $\mathrm{G}^2(\mathbf{p})\cap \mathrm{H}(\mathbf{p})\neq\emptyset$.} and (iii) if $\mathrm{G}^2(\mathbf{p})\neq\emptyset$ and $\mathrm{G}^2(\mathbf{p})\cap \mathrm{H}(\mathbf{p})=\emptyset$, we can repeat the above arguments to define $\Omega^3(\mathbf{p})$ and $\mathrm{G}^3(\mathbf{p})$ and implement the same examination. 

\subsection{Verifying equilibrium price vectors}\label{Sec_verify}

The above discussion implies the following procedure for verifying whether a balanced price vector $\mathbf{p}\in\mathbb{B}$ is an equilibrium price vector. This procedure identifies a stable outcome if $\mathbf{p}$ is an equilibrium price vector. Let $\Omega^1(\mathbf{p})\equiv\Omega$.

\bigskip

\noindent \textbf{A procedure for verifying equilibrium price vectors} 

\bigskip

Step $k,k\geq1$: 
\begin{description}
  \item[(i)] If $\mathrm{G}^k(\mathbf{p})=\emptyset$, then $(\bigcup_{i\in I}\overline{\mathrm{d}}^k_i(\mathbf{p}_i),\mathbf{p})$ is a competitive equilibrium in which $\bigcup_{i\in I}\overline{\mathrm{d}}^k_i(\mathbf{p}_i)$ is the largest efficient set of primitive contracts. The procedure terminates by producing the stable outcome $\{(w,\mathbf{t}^w)|w\in\bigcup_{i\in I}\overline{\mathrm{d}}^k_i(\mathbf{p}_i)\text{ and }t_i^w=p_i^w\text{ for each }i\in I\}$.
  \item[(ii)] If $\mathrm{G}^k(\mathbf{p})\cap \mathrm{H}(\mathbf{p})\neq\emptyset$, then $\mathbf{p}$ is a non-equilibrium price vector, and the procedure terminates.
  \item[(iii)] If $\mathrm{G}^k(\mathbf{p})\neq\emptyset$ and $\mathrm{G}^k(\mathbf{p})\cap \mathrm{H}(\mathbf{p})=\emptyset$, then define
\begin{equation}\label{def_proce}
\begin{aligned}
&\Omega^{k+1}(\mathbf{p})\equiv \Omega^k(\mathbf{p})\setminus \mathrm{G}^k(\mathbf{p}),\\
&\overline{\mathrm{d}}^{k+1}_i(\mathbf{p}_i)\equiv\bigcup_{\Psi^l\in \mathrm{D}_i(\mathbf{p}_i):\Psi^l\subseteq\Omega^{k+1}(\mathbf{p})}\Psi^l\qquad\qquad\text{ for each agent }i\in I,\\
&\mathrm{G}^{k+1}(\mathbf{p})\equiv\{w|w\in\overline{\mathrm{d}}^{k+1}_j(\mathbf{p}_j)\text{ and }w\notin\overline{\mathrm{d}}^{k+1}_{j'}(\mathbf{p}_{j'}) \text{ for some }j,j'\in \mathrm{N}(w)\},
\end{aligned}
\end{equation}
\end{description}
and go to the next step.
\bigskip

At each Step $k\geq2$, primitive contracts from $\mathrm{G}^k(\mathbf{p})$ are those ``partially agreeable'' when we only consider agents' demand sets confined in $\Omega^k(\mathbf{p})$, where primitive contracts from $\Omega^k(\mathbf{p})$ are those ``not partially agreeable'' at Step $k-1$. The set $\overline{\mathrm{d}}^{k}_i(\mathbf{p}_i)$ is agent $i$'s largest demand set confined in $\Omega^k(\mathbf{p})$. The procedure must terminate at some step since the primitive contracts are finite. If the procedure terminates at Step $s$, then for each $k\in\{1,\ldots,s\}$, we call primitive contracts from $\mathrm{G}^k(\mathbf{p})$ \textbf{level-$k$ partially agreeable} at $\mathbf{p}$.

At each step $k\geq1$, the procedure first examines whether there are primitive contracts partially agreeable of the current level: $\mathbf{p}$ is an equilibrium price vector if $\mathrm{G}^k(\mathbf{p})=\emptyset$. If $\mathrm{G}^k(\mathbf{p})\neq\emptyset$, the procedure examines whether there are strongly demanded primitive contracts partially agreeable of the current level: $\mathbf{p}$ is a non-equilibrium price vector if $\mathrm{G}^k(\mathbf{p})\cap \mathrm{H}(\mathbf{p})\neq\emptyset$. If $\mathrm{G}^k(\mathbf{p})\neq\emptyset$ and $\mathrm{G}^k(\mathbf{p})\cap \mathrm{H}(\mathbf{p})=\emptyset$ at Step $k$, the procedure removes level-$k$ partially agreeable primitive contracts and goes to the next step.

In the Appendix, we formally prove that the above procedure correctly verifies whether a balanced price vector $\mathbf{p}\in \mathbb{B}$ is an equilibrium price vector. Our analysis also implies that, under gross complementarity, a balanced price vector $\mathbf{p}\in \mathbb{B}$ is a non-equilibrium price vector if and only if there exists a strongly demanded primitive contract $w\in \mathrm{H}(\mathbf{p})$ that is partially agreeable of some level.

\begin{example}\label{exam_main}
\normalfont
There are three agents $i_1,i_2$, and $i_3$, and five primitive contracts $w_1,w_2,\cdots,w_5$. Agent $i_1$ and agent $i_2$ can sign $w_1,w_4$, and $w_5$; agent $i_1$ and agent $i_3$ can sign $w_2$ and $w_3$. Suppose primitive contracts are gross complements for all agents, and agents' demand sets at some price vector $\mathbf{p}\in \mathbb{B}$ is given by 
\begin{align*}
&\mathrm{D}_{i_1}(\mathbf{p})=\{\{w_1,w_2\},\{w_1,w_2,w_3,w_4\}\}, \quad \mathrm{D}_{i_2}(\mathbf{p})=\{\{w_1\},\{w_1,w_4,w_5\}\},\\ &\mathrm{D}_{i_3}(\mathbf{p})=\{\emptyset,\{w_2,w_3\}\}
\end{align*}
There are two strongly demanded primitive contracts: the primitive contract $w_1$ is strongly demanded by the agents $i_1$ and $i_2$, and the primitive contract $w_2$ is strongly demanded by the agent $i_1$. We have $\mathrm{H}(\mathbf{p})=\{w_1,w_2\}$. The procedure for verifying whether $\mathbf{p}$ is an equilibrium price vector proceeds as follows.

Step 1. There is one level-1 partially agreeable primitive contract $w_5$, which is demanded by the agent $i_2$ but not demanded by the agent $i_1$. We have $\mathrm{G}^1(\mathbf{p})=\{w_5\}$, and thus $\mathrm{G}^1(\mathbf{p})\cap \mathrm{H}(\mathbf{p})=\emptyset$. After removing $w_5$, we have $\Omega^2(\mathbf{p})=\{w_1,w_2,w_3,w_4\}$. The agents' largest demand sets confined in $\Omega^2(\mathbf{p})$ are $\overline{\mathrm{d}}^2_{i_1}(\mathbf{p}_{i_1})=\{w_1,w_2,w_3,w_4\}$, $\overline{\mathrm{d}}^2_{i_2}(\mathbf{p}_{i_2})=\{w_1\}$, and $\overline{\mathrm{d}}^2_{i_3}(\mathbf{p}_{i_3})=\{w_2,w_3\}$. There is one level-2 partially agreeable primitive contract $w_4$, and thus $\mathrm{G}^2(\mathbf{p})=\{w_4\}$.

Step 2. $\mathrm{G}^2(\mathbf{p})\cap\mathrm{H}(\mathbf{p})=\emptyset$. We have $\Omega^3(\mathbf{p})=\{w_1,w_2,w_3\}$, $\overline{\mathrm{d}}^3_{i_1}(\mathbf{p}_{i_1})=\{w_1,w_2\}$, $\overline{\mathrm{d}}^3_{i_2}(\mathbf{p}_{i_2})=\{w_1\}$, $\overline{\mathrm{d}}^3_{i_3}(\mathbf{p}_{i_3})=\{w_2,w_3\}$, and $\mathrm{G}^3(\mathbf{p})=\{w_3\}$.

Step 3. $\mathrm{G}^3(\mathbf{p})\cap\mathrm{H}(\mathbf{p})=\emptyset$. We have $\Omega^4(\mathbf{p})=\{w_1,w_2\}$, $\overline{\mathrm{d}}^4_{i_1}(\mathbf{p}_{i_1})=\{w_1,w_2\}$, $\overline{\mathrm{d}}^4_{i_2}(\mathbf{p}_{i_2})=\{w_1\}$, $\overline{\mathrm{d}}^4_{i_3}(\mathbf{p}_{i_3})=\emptyset$, and $\mathrm{G}^4(\mathbf{p})=\{w_2\}$.

Step 4. $\mathrm{G}^4(\mathbf{p})\cap\mathrm{H}(\mathbf{p})=\{w_2\}$. Hence, the price vector $\mathbf{p}$ is a non-equilibrium price vector. The procedure terminates.
\end{example}

\section{The dynamic auction\label{Sec_auction}}

Four assumptions are critical in the auction of this paper. The first one is the assumption of transferable utilities, and the second one is gross complementarity. We make two additional assumptions. The third one is that \textbf{agents' valuations take integer values}: $\mathrm{v}_i(\Psi)\in \mathbb{Z}$ for each agent $i\in I$ and $\Psi\subseteq\Omega_i$. This assumption is natural in real life since we cannot value a primitive contract more closely than to the nearest penny.

The last assumption is that \textbf{agents report truthfully} in the auction. We discuss on the strategic issues below. 

\subsection{Strategic issues}\label{Sec_strategy}

The following example shows that two agents collaborating together have incentives to underreport their benefits in a direct mechanism. 

\begin{example}
\normalfont
Consider a simple cooperation between Ana and Bob, who can sign a primitive contract $w$. We assume $\mathrm{v}_{Ana}(\{w\})=2$, $\mathrm{v}_{Bob}(\{w\})=3$, and $\mathrm{v}_{Ana}(\emptyset)=\mathrm{v}_{Bob}(\emptyset)=0$. The set of stable outcomes are
\begin{equation*}
\{\{(w,(t^w_{Ana},t^w_{Bob}))\}|t^w_{Ana}+t^w_{Bob}=0, t^w_{Ana}\leq 2,\text{ and }t^w_{Bob}\leq 3\}. 
\end{equation*} 
Consider a direct mechanism that always selects a stable outcome. If this mechanism selects an outcome $\{(w,(t^w_{Ana},t^w_{Bob}))\}$ with $t^w_{Ana}\geq0$ when both agents report their valuations truthfully, then Ana has an incentive to report $\mathrm{v}_{Ana}(\{w\})=a\in(-3,0)$ since the mechanism would instead choose an outcome from $\{\{(w,(t^w_{Ana},t^w_{Bob}))\}|t^w_{Ana}+t^w_{Bob}=0, t^w_{Ana}\leq a,\text{ and }t^w_{Bob}\leq 3\}$. Similarly, if this mechanism selects an outcome $\{(w,(t^w_{Ana},t^w_{Bob}))\}$ with $t^w_{Ana}<0$ (and thus $t^w_{Bob}>0$) when both agents report their valuations truthfully, then Bob has an incentive to report $\mathrm{v}_{Bob}(\{w\})\in(-2,0]$. 
\end{example}

According to the revelation principle, any problem subsuming this example does not admit an incentive-compatible mechanism that always selects a stable outcome. Clearly, the agents' valuations in this example are supermodular, and thus agents' preferences satisfy gross complementarity. As a consequence, there is no incentive-compatible mechanism that always produces a stable outcome in the problem we study.

\subsection{Complements chain}\label{Sec_CC}

Suppose the procedure for verifying whether $\mathbf{p}\in\mathbb{B}$ is an equilibrium price vector terminates at the first step as $\mathrm{G}^1(\mathbf{p})\cap \mathrm{H}(\mathbf{p})\neq\emptyset$. We then know that $\mathbf{p}$ is a non-equilibrium price vector. Pick a primitive contract $w\in \mathrm{G}^1(\mathbf{p})\cap \mathrm{H}(\mathbf{p})$. Suppose $w$ is strongly demanded by agent $i$ but not demanded by agent $j$, then one can decrease the value of the Lyapunov function $\mathcal{L}(\mathbf{p})$ by both increasing $p^w_i$ by one unit and decreasing $p^w_j$ by one unit. Under the new price vector, since agent $i$'s valuation takes integer values, agent $i$'s indirect utility decreases by one unit,\footnote{This is because (i) each of agent $i$'s optimal choices at $\mathbf{p}$ yields a utility $\mathrm{V}_i(\mathbf{p}_i)-1$ at the new price vector, and (ii) each of agent $i$'s suboptimal choices at $\mathbf{p}$ yields at most a utility $\mathrm{V}_i(\mathbf{p}_i)-1$ at both $\mathbf{p}$ and the new price vector.} and agent $j$'s indirect utility remains unchanged.\footnote{This is because (i) each of agent $j$'s optimal choices at $\mathbf{p}$ yields the same utility $\mathrm{V}_j(\mathbf{p}_j)$ at $\mathbf{p}$ and the new price vector, and (ii) each of agent $j$'s suboptimal choices at $\mathbf{p}$ yields at most $\mathrm{V}_j(\mathbf{p}_j)-1$ at $\mathbf{p}$ and at most $\mathrm{V}_j(\mathbf{p}_j)$ at the new price vector.} Other agents' indirect utilities are also unchanged. Therefore, the value of the Lyapunov function $\mathcal{L}(\mathbf{p})$ decreases by one unit.

More generally, when the procedure for verifying whether $\mathbf{p}\in\mathbb{B}$ is an equilibrium price vector terminates at Step $s\geq1$ as $\mathrm{G}^s(\mathbf{p})\cap \mathrm{H}(\mathbf{p})\neq\emptyset$, we want to identify a structure called complements chain. We call a primitive contract $w'\in \Omega_i$ a \textbf{complement} to another primitive contract $w\in \Omega_i$ for agent $i$ at a price vector $\mathbf{p}_i\in\mathbb{R}^{\Omega_i}$ if (i) $w\in\Psi$ for some $\Psi\in \mathrm{D}_i(\mathbf{p}_i)$, and (ii) $w\in\Psi$ and $\Psi\in \mathrm{D}_i(\mathbf{p}_i)$ imply $w'\in\Psi$. That is, agent $i$ chooses the primitive contract $w$ at price vector $\mathbf{p}_i$ only when she also chooses the primitive contract $w'$. For each primitive contract $w$ that belongs to some demand set of agent $i$ at $\mathbf{p}_i$, let $\underline{\mathrm{d}}_i^w(\mathbf{p}_i)$ be the intersection of all agent $i$'s demand sets that contain $w$. That is,
\begin{equation*}
\underline{\mathrm{d}}_i^w(\mathbf{p}_i)\equiv\bigcap_{\Psi^j\in \mathrm{D}_i(\mathbf{p}_i)\,:\, w\in\Psi^j}\Psi^j \qquad\text{ where }w\in\Psi\text{ for some }\Psi\in \mathrm{D}_i(\mathbf{p}_i)
\end{equation*}
Notice that $w'$ is a complement to $w$ for agent $i$ at $\mathbf{p}_i$ if and only if $w'\in\underline{\mathrm{d}}_i^w(\mathbf{p}_i)$. Moreover, when primitive contracts are gross complements for agent $i$, the set $\underline{\mathrm{d}}_i^w(\mathbf{p}_i)$ is also agent $i$'s demand set at $\mathbf{p}_i$.

Suppose at $\mathbf{p}$ a primitive contract $w\in G^k(\mathbf{p})$ is level-$k\geq2$ partially agreeable as it becomes not demanded by some agent $j\in \mathrm{N}(w)$ at Step $k$ of the procedure. The primitive contract $w$ is not ``partially agreeable'' at Step $k-1$ and belongs to at least one of $j$'s demand sets confined in $\Omega^{k-1}$. Since agent $j$ does not demand $w$ at Step $k$, the set $\underline{\mathrm{d}}_j^w(\mathbf{p}_j)$ does not belong to any of agent $j$'s demand sets confined in $\Omega^k$. Thus, there must be a primitive contract $w'\in\underline{\mathrm{d}}_j^w(\mathbf{p}_j)$ removed at Step $k-1$. Clearly, the primitive contract $w'$ belongs to $\mathrm{G}^{k-1}(\mathbf{p})\cap\overline{\mathrm{d}}^{k-1}_j(\mathbf{p}_j)$ and is a complement to $w$ for agent $j$ at $\mathbf{p}_j$. We have the following lemma.

\begin{lemma}\label{lma_complements}
\normalfont
If the procedure for verifying whether $\mathbf{p}\in\mathbb{B}$ is an equilibrium price vector stops at Step $s\geq2$, then for each $k\in\{2,\cdots,s\}$, each level-$k$ partially agreeable primitive contract $w\in \mathrm{G}^k(\mathbf{p})$, and each agent $j\in \mathrm{N}(w)$ such that $w\notin\overline{\mathrm{d}}^k_j(\mathbf{p}_j)$, there is a level-$(k-1)$ partially agreeable primitive contract $w'\in \mathrm{G}^{k-1}(\mathbf{p})\cap\overline{\mathrm{d}}^{k-1}_j(\mathbf{p}_j)$ that is a complement to $w$ for agent $j$.
\end{lemma}

Lemma \ref{lma_complements} indicates a structure called ``complements chain'' at a non-equilibrium price vector.

\begin{definition}\label{def_chain}
\normalfont
If the procedure for verifying whether $\mathbf{p}\in\mathbb{B}$ is an equilibrium price vector stops at Step $s\geq1$ as $\mathrm{G}^s(\mathbf{p})\cap \mathrm{H}(\mathbf{p})\neq\emptyset$, a chain $i^1w^1i^2w^2\cdots w^si^{s+1}$ with $w^l\in \Omega_{i^l}\cap \Omega_{i^{l+1}}$ for all $1\leq l\leq s$ is called a \textbf{complements chain} at $\mathbf{p}\in \mathbb{B}$ if 
\begin{description}
  \item[(i)] $w^l\in \mathrm{G}^l(\mathbf{p})$, $w^l\notin\overline{\mathrm{d}}^l_{i^l}(\mathbf{p}_{i^l})$, and $w^l\in\overline{\mathrm{d}}^l_{i^{l+1}}(\mathbf{p}_{i^{l+1}})$ for all $1\leq l\leq s$; 
  \item[(ii)] $w^s\in \mathrm{H}(\mathbf{p})$ and $w^s\in\underline{\mathrm{d}}_{i^{s+1}}(\mathbf{p}_{i^{s+1}})$; and
  \item[(iii)] if $s\geq2$, then $w^{l-1}$ is a complement to $w^l$ for agent $i^l$ at $\mathbf{p}_{i^l}$ for all $2\leq l\leq s$.
\end{description}
\end{definition}

A complements chain at $\mathbf{p}\in \mathbb{B}$ starts with an agent $i^1$ and a (level-1) partially agreeable primitive contract $w^1$ that is not demanded by $i^1$ at $\mathbf{p}_{i^1}$. The chain ends with an agent $i^{s+1}$ who strongly demands $w^s$ at $\mathbf{p}_{i^{s+1}}$. The $l$th contract in the chain is level-$l$ partially agreeable as $w^l\notin\overline{\mathrm{d}}^l_{i^l}(\mathbf{p}_{i^l})$ and $w^l\in\overline{\mathrm{d}}^l_{i^{l+1}}(\mathbf{p}_{i^{l+1}})$. If the chain contains more than one primitive contract, then the $(l-1)$th one is a complement to the $l$th one for the agent between the two primitive contracts at the current price vector for all $2\leq l\leq s$. Notice that there are actually no complements in a complements chain with only one primitive contract (i.e., $s=1$). In a complements chain, the primitive contracts are distinct since the sets of partially agreeable contracts of different levels are disjoint, but an agent may appear multiple times in the chain. However, two neighbour agents $i^l$ and $i^{l+1}$ are distinct since $w^l\notin\overline{\mathrm{d}}^l_{i^l}(\mathbf{p}_{i^l})$ and $w^l\in\overline{\mathrm{d}}^l_{i^{l+1}}(\mathbf{p}_{i^{l+1}})$.

Suppose there is a complements chain $i^1w^1i^2w^2\cdots w^si^{s+1}$ at a balanced price vector $\mathbf{p}\in \mathbb{B}$. Then, for any competitive equilibrium $(\mathbf{p},\Phi)$ at $\mathbf{p}$, we have $w^s\in\Phi$ since the agent $i^{s+1}$ strongly demands $w^s$ at $\mathbf{p}_{i^{s+1}}$, and $w^1\notin\Phi$ since the agent $i^1$ does not want to sign $w^1$ at $\mathbf{p}_{i^1}$. However, the complements chain at $\mathbf{p}$ indicates that the agent $i^s$ signs $w^s$ at $\mathbf{p}$ only if she also signs $w^{s-1}$; and the agent $i^{s-1}$ signs $w^{s-1}$ at $\mathbf{p}$ only if she also signs $w^{s-2}$; and so on. These relations imply that the primitive contract $w^s$ can be signed at $\mathbf{p}$ only if the primitive contract $w^1$ is also signed. This contradicts that the agent $i^1$ does not want to sign $w^1$ at $\mathbf{p}_{i^1}$. Therefore, we know that a balanced price vector $\mathbf{p}\in \mathbb{B}$ is a non-equilibrium price vector if there exists a complements chain at $\mathbf{p}$. The converse of this statement is also true.

\begin{proposition}\label{prop_chain}
\normalfont
When primitive contracts are gross complements for all agents, a balanced price vector $\mathbf{p}\in \mathbb{B}$ is a non-equilibrium price vector if and only if there exists a complements chain at $\mathbf{p}$.
\end{proposition}

If $\mathbf{p}\in \mathbb{B}$ is a non-equilibrium price vector,
the procedure for verifying whether $\mathbf{p}$ is an equilibrium price vector must stop at some Step $s$ as $\mathrm{G}^s(\mathbf{p})\cap \mathrm{H}(\mathbf{p})\neq\emptyset$. Then we can obtain a complements chain as follows. Pick a primitive contract $w^s$ from $\mathrm{G}^s(\mathbf{p})\cap \mathrm{H}(\mathbf{p})$ and let $i^{s+1}$ be an agent who strongly demands $w^s$ (i.e., $w^s\in\underline{\mathrm{d}}_{i^{s+1}}(\mathbf{p}_i)$); we thus have $w^s\in\overline{\mathrm{d}}^s_{i^{s+1}}(\mathbf{p}_{i^{s+1}})$.\footnote{Since all primitive contracts from $\underline{\mathrm{d}}_{i^{s+1}}(\mathbf{p}_i)$ are in $\mathrm{H}(\mathbf{p})$ and thus are still in $\Omega^s(\mathbf{p})$, we have $\underline{\mathrm{d}}_{i^{s+1}}(\mathbf{p}_i)\subseteq\overline{\mathrm{d}}^s_{i^{s+1}}(\mathbf{p}_{i^{s+1}})$ and thus $w^s\in\overline{\mathrm{d}}^s_{i^{s+1}}(\mathbf{p}_{i^{s+1}})$.} Let $i^s\in \mathrm{N}(w^s)$ be an agent who does not demand $w^s$ at Step $s$ (i.e., $w^s\notin\overline{\mathrm{d}}^s_{i^s}(\mathbf{p}_i)$). If $s=1$, we obtain a complements chain $i^1w^1i^2$. If $s\geq2$, by Lemma \ref{lma_complements} we can choose a primitive contract $w^{s-1}\in \mathrm{G}^{s-1}(\mathbf{p})\cap\overline{\mathrm{d}}^{s-1}_{i^s}(\mathbf{p}_{i^s})$ that is a complement to $w^s$ for agent $i^s$. We then continue to choose an agent $i^{s-1}\in \mathrm{N}(w^{s-1})$ who does not demand $w^{s-1}$ at Step $s-1$: $w^{s-1}\notin\overline{\mathrm{d}}^{s-1}_{i^{s-1}}(\mathbf{p}_{i^{s-1}})$, and, if $s\geq3$, a primitive contract $w^{s-2}\in \mathrm{G}^{s-2}(\mathbf{p})\cap\overline{\mathrm{d}}^{s-2}_{i^{s-1}}(\mathbf{p}_{i^{s-1}})$ that is a complement to $w^{s-1}$ for agent $i^{s-1}$. We can repeat this construction backward the procedure until we reach a level-1 partially agreeable primitive contract $w^1$ and an agent $i^1\in \mathrm{N}(w^1)$ who does not demand $w^1$ at $\mathbf{p}_{i^1}$. For instance, using this method we can find a complements chain $i_1w_5i_2w_4i_1w_3i_3w_2i_1$ in Example \ref{exam_main}. 

\subsection{Price adjustment at a non-equilibrium price vector}\label{Sec_adjust}

Once we find a complements chain at a non-equilibrium price vector $\mathbf{p}$, we can adjust this price vector according to the following lemma so as to decrease the value of the Lyapunov function $\mathcal{L}(\mathbf{p})$.

\begin{lemma}\label{lma_adjust}
\normalfont
Suppose the procedure for verifying whether $\mathbf{p}\in\mathbb{B}$ is an equilibrium price vector stops at Step $s\geq1$ as $\mathrm{G}^s(\mathbf{p})\cap \mathrm{H}(\mathbf{p})\neq\emptyset$, and $i^1w^1i^2w^2\cdots w^si^{s+1}$ is a complements chain. Let $\mathbf{p}'\in \mathbb{B}$ be a balanced price vector such that
\begin{equation}\label{newprice}
p^{'w^l}_{i^l}=p^{w^l}_{i^l}-1 \quad\text{ and }\quad p^{'w^l}_{i^{l+1}}=p^{w^l}_{i^{l+1}}+1 \qquad\text{ for all }1\leq l\leq s
\end{equation} 
and other coordinates of $\mathbf{p}'$ are the same as those of $\mathbf{p}$, then we have $\mathrm{V}_i(\mathbf{p}'_i)=\mathrm{V}_i(\mathbf{p}_i)$ for all $i\in\{i^1,\ldots,i^s\}\setminus\{i^{s+1}\}$, and $\mathrm{V}_{i^{s+1}}(\mathbf{p}'_{i^{s+1}})=\mathrm{V}_{i^{s+1}}(\mathbf{p}_{i^{s+1}})-1$.
\end{lemma}

The adjustment (\ref{newprice}) means that for each primitive contract in the complements chain, the price for its left-hand-side agent increases by one unit, and the price for its right-hand-side agent decreases by one unit. If an agent is not the last agent in the chain, her indirect utilities are the same at $\mathbf{p}$ and $\mathbf{p}'$. For the last agent in the chain, her indirect utility decreases by one unit as the price vector changes. This lemma implies that, when we find a complements chain, we can adjust the price vector according to (\ref{newprice}) so as to decrease the value of the Lyapunov function by one unit.

If an agent $i\in\{i^1,\ldots,i^s\}\setminus\{i^{s+1}\}$ appears in the complements chain and is not the last agent of the chain, then for any $\Psi\in \mathrm{D}_i(\mathbf{p}_i)$ that contains some of her right-hand-side primitive contracts in the chain, these primitive contracts cost less at the new price vector $\mathbf{p}'$. But according to the condition (iii) of Definition \ref{def_chain}, her left-hand-side primitive contracts that goes right before these right-hand-side primitive contracts must also be in $\Psi$. Hence, any of her demand sets at $\mathbf{p}$ does not cost less at the new price vector $\mathbf{p}'$ and thus  does not yield a higher utility at $\mathbf{p}'$: $\mathrm{U}_i(\Psi,\mathbf{p}'_i)\leq \mathrm{V}_i(\mathbf{p}_i)$ for any $\Psi\in \mathrm{D}_i(\mathbf{p}_i)$. For the last agent $i^{s+1}$ of the complements chain, we instead have $\mathrm{U}_{i^{s+1}}(\Psi,\mathbf{p}'_{i^{s+1}})\leq \mathrm{V}_{i^{s+1}}(\mathbf{p}_{i^{s+1}})-1$ for any $\Psi\in \mathrm{D}_{i^{s+1}}(\mathbf{p}_{i^{s+1}})$ since the last primitive contract is strongly demanded by $i^{s+1}$ and costs one unit more at the new price vector. Moreover, the largest demand set $\overline{\mathrm{d}}^1_i(\mathbf{p}_i)$ of agent $i\in\{i^1,\ldots,i^s\}\setminus\{i^{s+1}\}$ at $\mathbf{p}$ has the same cost as at $\mathbf{p}'$; and agent $i^{s+1}$'s largest demand set $\overline{\mathrm{d}}^1_{i^{s+1}}(\mathbf{p}_{i^{s+1}})$ at $\mathbf{p}$ cost one unit more at $\mathbf{p}'$ than at $\mathbf{p}$.\footnote{This is because all primitive contracts of the complements chain involving agent $i\in I$ are in $\overline{\mathrm{d}}^1_i(\mathbf{p}_i)$ except that the primitive contract $w^1$ is not in $\overline{\mathrm{d}}^1_i(\mathbf{p}_i)$ if $i=i^1$. Then, this statement follows from the condition (iii) of Definition \ref{def_chain}.\label{ft_ds}} Therefore, we have $\mathrm{U}_i(\overline{\mathrm{d}}^1_i(\mathbf{p}_i),\mathbf{p}'_i)=\mathrm{V}_i(\mathbf{p}_i)$ for all $i\in\{i^1,\ldots,i^s\}\setminus\{i^{s+1}\}$, and $\mathrm{U}_{i^{s+1}}(\overline{\mathrm{d}}^1_{i^{s+1}}(\mathbf{p}),\mathbf{p}'_{i^{s+1}})=\mathrm{V}_{i^{s+1}}(\mathbf{p}_{i^{s+1}})-1$.

Now, we can prove Lemma \ref{lma_adjust} by showing that for any agent $i\in\{i^1,\ldots,i^s\}\setminus\{i^{s+1}\}$, any of her suboptimal choices $\Psi\notin \mathrm{D}_i(\mathbf{p}_i)$ at $\mathbf{p}$ does not yield a utility higher than $\mathrm{V}_i(\mathbf{p}_i)$ at the new price vector $\mathbf{p}'$: $\mathrm{U}_i(\Psi,\mathbf{p}'_i)\leq \mathrm{V}_i(\mathbf{p}_i)$ for any $\Psi\in 2^{\Omega_i}\setminus \mathrm{D}_i(\mathbf{p}_i)$; and for the last agent $i^{s+1}$ of the complements chain, any of her suboptimal choices $\Psi\notin \mathrm{D}_{i^{s+1}}(\mathbf{p}_{i^{s+1}})$ at $\mathbf{p}$ does not yield a utility higher than $\mathrm{V}_{i^{s+1}}(\mathbf{p}_{i^{s+1}})-1$ at the new price vector $\mathbf{p}'$: $\mathrm{U}_{i^{s+1}}(\Psi,\mathbf{p}'_{i^{s+1}})\leq \mathrm{V}_{i^{s+1}}(\mathbf{p}_{i^{s+1}})-1$ for any $\Psi\in 2^{\Omega_{i^{s+1}}}\setminus \mathrm{D}_{i^{s+1}}(\mathbf{p}_{i^{s+1}})$. We prove this part in the Appendix.

\subsection{The formal procedure of the auction}

At a non-equilibrium price vector $\mathbf{p}$, there is at least one complements chain. The auctioneer can identify a complements chain (using the method described below Proposition \ref{prop_chain}) and adjust the price vector by applying (\ref{newprice}). According to Lemma \ref{lma_adjust}, this would decrease the value of the Lyapunov function $\mathcal{L}(\mathbf{p})$ by one unit. However, the auctioneer may find several complements chains with \textbf{disjoint sets of agents}: The auctioneer can identify the first complements chain, and then try to construct a second one that contains no agent from the first one;\footnote{The auctioneer constructs the second complements chain using the same method as she constructs the first one as long as she does not select agents that appear in the first complements chain.} the auctioneer continues to identify complements chains with disjoint sets of agents until she finds it impossible to construct a new complements chain that contains no agent from previous ones. Then the auctioneer adjusts the price vector $\mathbf{p}$ by applying (\ref{newprice}) for each complements chain. This would decrease the value of the Lyapunov function $\mathcal{L}(\mathbf{p})$ by the number of the complements chains the auctioneer identifies. 

We are now ready to provide a formal description of a dynamic auction for multilateral collaboration with transferable utilities and gross-complement primitive contracts. We can start the auction with the price vector $\mathbf{0}$, which means that agents do not pay or receive money in all primitive contracts. Notice that we can also start with any other balanced integer price vector $\mathbf{p}\in \mathbb{B}\cap\mathbb{Z}^n$.

\bigskip

\noindent \textbf{A dynamic auction for multilateral collaboration}

Step 0: The auctioneer announces an initial integer balanced price vecotr $\mathbf{p}(1)\in \mathbb{B}\cap\mathbb{Z}^n$. 

Step $s, s\geq1$: Each agent reports all her demand sets from $\mathrm{D}_i(\mathbf{p}(s))$. The auctioneer uses the procedure in Section \ref{Sec_verify} to verify whether $\mathbf{p}(s)$ is an equilibrium price vector. If $\mathbf{p}(s)$ is an equilibrium price vector, the procedure in Section \ref{Sec_verify} outputs a stable outcome, and the auction terminates. If $\mathbf{p}(s)$ is a non-equilibrium price vector, the auctioneer finds complements chains of disjoint sets of agents, adjusts $\mathbf{p}(s)$ into a new price vector $\mathbf{p}(s+1)$ by applying (\ref{newprice}) for each complements chain, and goes to the next step.

\begin{theorem}
\normalfont
If agents' valuations take integer values, primitive contracts are gross complements for all agents, and agents report truthfully, then the dynamic auction for multilateral collaboration converges to a stable outcome in a finite number of steps.
\end{theorem}

Recall that the valuation $\mathrm{v}_i(\Psi)$ is finite for any agent $i\in I$ and any set of primitive contracts $\Psi\subseteq\Omega_i$. We thus know that the indirect utility $\mathrm{V}_i(\Psi)$ and the value of the Lyapunov function $\mathcal{L}(\mathbf{p})$ are finite for any agent $i\in I$ and any balanced price vector $\mathbf{p}\in \mathbb{B}$. Consequently, the auction converges to a stable outcome in a finite number of steps as the value of the Lyapunov function $\mathcal{L}(\mathbf{p})$ is reduced by at least one unit in each step of the auction. 

RY showed gross complementarity in a market of patent licensing,\footnote{See Example 3 of RY.} making our auction a suitable tool for finding a stable outcome in this market. In an economy with income effects, \cite{BEJK23} demonstrated that monotone auctions may fail to find an equilibrium when gross substitutability does not hold. Although ascending-price auctions have been successfully applied to allocate spectrum licenses in real-life practice, designers often prefer combinatorial auctions when complementarities are strong. However, our work shows that, when allowing prices to adjust in both upward and downward directions, dynamic auctions are applicable to environments in which complementarities are a key feature. Our analysis may provide insights into the design of dynamic auctions for environments with other forms of complementarities.

\bigskip

\bigskip

\section{Appendix}

\subsection{Proofs of the lemmata}\label{Sec_proof}

\begin{proof}[Proof of Lemma \ref{lma_effi}]
(i) Let $\Psi\subseteq\Omega$ be an arbitrary set of primitive contracts. If $(\Phi,\mathbf{p})$ is a competitive equilibrium, then $\mathrm{U}_i(\Phi_i,\mathbf{p}_i)\geq \mathrm{U}_i(\Psi_i,\mathbf{p}_i)$ for each $i\in I$, and thus, $\sum_{i\in I}\mathrm{v}_i(\Phi_i)=\sum_{i\in I}\mathrm{U}_i(\Phi_i,\mathbf{p}_i)\geq\sum_{i\in I}\mathrm{U}_i(\Psi_i,\mathbf{p}_i)=\sum_{i\in I}\mathrm{v}_i(\Psi_i)$, where the two equalities hold since $\mathbf{p}$ is balanced. Hence, $\Phi$ is efficient.

(ii) Suppose $(\Phi',\mathbf{p})$ is a competitive equilibrium, and let $\Phi\subseteq\Omega$ be an arbitrary efficient set of primitive contracts; then, since $\mathrm{U}_i(\Phi'_i,\mathbf{p}_i)\geq \mathrm{U}_i(\Phi_i,\mathbf{p}_i)$ for each $i\in I$, we have $\sum_{i\in I}\mathrm{v}_i(\Phi_i)\geq\sum_{i\in I}\mathrm{v}_i(\Phi'_i)=\sum_{i\in I}\mathrm{U}_i(\Phi'_i,\mathbf{p}_i)\geq \sum_{i\in I}\mathrm{U}_i(\Phi_i,\mathbf{p}_i)=\sum_{i\in I}\mathrm{v}_i(\Phi_i)$, where the first inequality holds since $\Phi$ is efficient, the two equalities hold since $\mathbf{p}$ is balanced, and the second inequality holds since $(\Phi',\mathbf{p})$ is a competitive equilibrium. Therefore, $\mathrm{U}_i(\Phi'_i,\mathbf{p}_i)=\mathrm{U}_i(\Phi_i,\mathbf{p}_i)$ for each $i\in I$, and thus, $(\Phi,\mathbf{p})$ is also a competitive equilibrium. 
\end{proof}

\bigskip

\begin{proof}[Proof of Lemma \ref{lma_chara}]
The ``only if'' part. Suppose $\mathbf{p}'\in \mathbb{B}$ is an equilibrium price vector, and $\Phi^*$ an aribitrary efficient set of primitive contracts. By Lemma \ref{lma_effi}, $(\Phi^*,\mathbf{p}')$ is a competitive equilibrium, and thus, $\mathcal{L}(\mathbf{p}')=\sum_{i\in I}\mathrm{V}_i(\mathbf{p}'_i)=\sum_{i\in I}\mathrm{U}_i(\Phi^*_i,\mathbf{p}'_i)=\sum_{i\in I}\mathrm{v}_i(\Phi^*_i)=\max_{\Phi\subseteq\Omega}\sum_{i\in I}\mathrm{v}_i(\Phi_i)$.\footnote{The second equality holds since $(\Phi^*,\mathbf{p}')$ is a competitive equilibrium; and the third equality holds since $\mathbf{p}'$ is balanced.} For any $\mathbf{p}\in \mathbb{B}$, we have $\mathrm{V}_i(\mathbf{p}_i)\geq \mathrm{v}_i(\Phi^*_i)-\sum_{w\in\Phi^*_i}p_i^w$ for each $i\in I$, and thus, $\mathcal{L}(\mathbf{p})=\sum_{i\in I}\mathrm{V}_i(\mathbf{p}_i)\geq\sum_{i\in I}\mathrm{v}_i(\Phi^*_i)=\mathcal{L}(\mathbf{p}')$, which indicates $\mathbf{p}'\in\arg\min_{\mathbf{p}\in \mathbb{B}}\mathcal{L}(\mathbf{p})$.

The ``if'' part. Suppose $\mathbf{p}'\in\arg\min_{\mathbf{p}\in \mathbb{B}}\mathcal{L}(\mathbf{p})$ and $\mathcal{L}(\mathbf{p}')=\max_{\Phi\subseteq\Omega}\sum_{i\in I}\mathrm{v}_i(\Phi_i)$. Let $\Phi^*$ be an arbitrary efficient set of primitive contracts. We have $\mathrm{V}_i(\mathbf{p}'_i)\geq \mathrm{v}_i(\Phi^*_i)-\sum_{w\in\Phi^*_i}p_i^{'w}$ for each $i\in I$. Suppose $\mathrm{V}_j(\mathbf{p}'_i)> \mathrm{v}_j(\Phi^*_j)-\sum_{w\in\Phi^*_j}p_j^{'w}$ for some $j\in I$, adding the inequalities for all agents we have $\mathcal{L}(\mathbf{p}')>\sum_{i\in I}\mathrm{v}_i(\Phi^*_i)$, which contradicts $\mathcal{L}(\mathbf{p}')=\max_{\Phi\subseteq\Omega}\sum_{i\in I}\mathrm{v}_i(\Phi_i)$. Therefore, $\mathrm{V}_i(\mathbf{p}'_i)=\mathrm{v}_i(\Phi^*_i)-\sum_{w\in\Phi^*_i}p_i^{'w}$ for each $i\in I$, and thus $(\Phi^*,\mathbf{p}')$ is a competitive equilibrium.
\end{proof}

\bigskip

\begin{proof}[Proof of Lemma \ref{lma_adjust}]
According to the discussion below Lemma \ref{lma_adjust}, it remains to show that (i) for any agent $i\in\{i^1,\ldots,i^s\}\setminus\{i^{s+1}\}$, any of her suboptimal choices $\Psi\notin \mathrm{D}_i(\mathbf{p}_i)$ at $\mathbf{p}$ does not yield a utility higher than $\mathrm{V}_i(\mathbf{p}_i)$ at the new price vector $\mathbf{p}'$: $\mathrm{U}_i(\Psi,\mathbf{p}'_i)\leq \mathrm{V}_i(\mathbf{p}_i)$ for any $\Psi\in 2^{\Omega_i}\setminus \mathrm{D}_i(\mathbf{p}_i)$, and (ii) for the last agent $i^{s+1}$ of the complements chain, any of her suboptimal choices $\Psi\notin \mathrm{D}_{i^{s+1}}(\mathbf{p}_{i^{s+1}})$ at $\mathbf{p}$ does not yield a utility higher than $\mathrm{V}_{i^{s+1}}(\mathbf{p}_{i^{s+1}})-1$ at the new price vector $\mathbf{p}'$: $\mathrm{U}_{i^{s+1}}(\Psi,\mathbf{p}'_{i^{s+1}})\leq \mathrm{V}_{i^{s+1}}(\mathbf{p}_{i^{s+1}})-1$ for any $\Psi\in 2^{\Omega_{i^{s+1}}}\setminus \mathrm{D}_{i^{s+1}}(\mathbf{p}_{i^{s+1}})$. 

Let $E=\{1,2,\cdots,s+1\}$ be the set of agents' superscripts in the complements chain. For each agent $i$ involved in the complements chain and each set of primitive contracts $\Psi\subseteq\Omega_i$ signed by agent $i$, let $f_i(\Psi)$ be the subset of $E$ such that $k\in f_i(\Psi)$ if the following three conditions hold: (i) $i^k=i$, (ii) $w^k\in \Psi$ if $k\leq s$, and (iii) $w^{k-1}\notin\Psi$ if $k\geq 2$. For each agent $i$, we have $f_i(\Psi)=\emptyset$ if $\Psi\in \mathrm{D}_i(\mathbf{p}_i)$ according to Condition (iii) of Definition \ref{def_chain} and the following two facts: (i) for agent $i^1$, $w^1\notin\Psi$ for any $\Psi\in \mathrm{D}_{i^1}(\mathbf{p}_{i^1})$; and (ii) for agent $i^{s+1}$, $w^s\in\Psi$ for all $\Psi\in \mathrm{D}_{i^{s+1}}(\mathbf{p}_{i^{s+1}})$. 

Hence, $|f_i(\Psi)|=1$ implies $\Psi\notin \mathrm{D}_i(\mathbf{p}_i)$, which further implies $\mathrm{U}_i(\Psi,\mathbf{p}_i)\leq \mathrm{V}_i(\mathbf{p}_i)-1$. For each agent $i\in\{i^1,\ldots,i^s\}\setminus\{i^{s+1}\}$, $|f_i(\Psi)|=1$ implies that the cost for the primitive contracts from $\Psi$ at $\mathbf{p}'$ is no less than the cost at $\mathbf{p}$ minus one.\footnote{If $f_i(\Psi)=\{k\}$ where $k\neq s+1$, by the definition of $f_i(\Psi)$ we have $i=i^k$. Notice that (i) the price of $w^k\in\Psi$ for agent $i$ is one unit lower at $\mathbf{p}'$ than at $\mathbf{p}$, (ii) $w^{k-1}$ is either not in $\Psi$ or non-existent (when $k=1$), and (iii) agent $i$ may be in other positions of the complements chain, and $\Psi$ may contain some of agent $i$'s left-hand-side primitive contracts of these positions. Therefore, the cost of $\Psi$ at $\mathbf{p}'$ is no less than its cost at $\mathbf{p}$ minus one.\label{ft_proof}} Thus, we have $\mathrm{U}_i(\Psi,\mathbf{p}'_i)\leq\mathrm{U}_i(\Psi,\mathbf{p}_i)+1\leq \mathrm{V}_i(\mathbf{p}_i)$ when $|f_i(\Psi)|=1$. 
Suppose $|f_{i^{s+1}}(\Psi)|=1$ for the last agent $i^{s+1}$ in the complements chain. There are two possible cases: (i) If $f_{i^{s+1}}(\Psi)\neq\{s+1\}$, then $w^s\in\Psi$; and thus, the cost for the primitive contracts from $\Psi\setminus\{w^s\}$ is at most one unit cheaper than at $\mathbf{p}$,\footnote{This statement follows the reason in footnote \ref{ft_proof}.} but the primitive contract $w^s$ costs one more unit; consequently, the cost for the primitive contracts from $\Psi$ at $\mathbf{p}'$ is no less than at $\mathbf{p}$; (ii) If $f_{i^{s+1}}(\Psi)=\{s+1\}$, then $w^s\notin\Psi$; and thus, the cost for the primitive contracts from $\Psi$ at $\mathbf{p}'$ is no less than at $\mathbf{p}$.\footnote{This is because agent $i^{s+1}$ may be in other positions of the complements chain, and $\Psi$ may contain some of agent $i^{s+1}$'s left-hand-side primitive contracts of these positions.} In both cases, we have $\mathrm{U}_{i^{s+1}}(\Psi,\mathbf{p}'_{i^{s+1}})\leq\mathrm{U}_{i^{s+1}}(\Psi,\mathbf{p}_{i^{s+1}})\leq \mathrm{V}_{i^{s+1}}(\mathbf{p}_{i^{s+1}})-1$. We have the following lemma if $|f_i(\Psi)|\geq2$ for any agent $i$ of the complements chain and any subset $\Psi\subseteq\Omega_i$.

\begin{lemma}\label{lma_phi}
\normalfont
For any agent $i\in\{i^1,\ldots,i^{s+1}\}$ of the complements chain and any subset $\Psi\subseteq\Omega_i$, if $|f_i(\Psi)|\geq2$, there exists $\Phi\in \mathrm{D}_i(\mathbf{p}_i)$ such that $|f_i(\Psi\cup \Phi)|,|f_i(\Psi\cap \Phi)|\leq|f_i(\Psi)|-1$ and
\begin{description}
  \item[(i)] $\mathrm{U}_i(\Phi,\mathbf{p}'_i)=\mathrm{V}_i(\mathbf{p}_i)$ if $i\in\{i^1,\ldots,i^s\}\setminus\{i^{s+1}\}$,
  \item[(ii)] $\mathrm{U}_i(\Phi,\mathbf{p}'_i)=\mathrm{V}_i(\mathbf{p}_i)-1$ if $i=i^{s+1}$.
\end{description}
\end{lemma}

\begin{proof}
Suppose $k,l\in f_i(\Psi)$ and $k>l$. We have $k-l\geq2$ since two neighbour agents in a complements chain are distinct. Let $\Phi=\overline{\mathrm{d}}^{k-1}_i(\mathbf{p}_i)$, we know $\mathrm{U}_i(\Phi,\mathbf{p}'_i)=\mathrm{V}_i(\mathbf{p}_i)$ if $i\in\{i^1,\ldots,i^s\}\setminus\{i^{s+1}\}$, and $\mathrm{U}_i(\Phi,\mathbf{p}'_i)=\mathrm{V}_i(\mathbf{p}_i)-1$ if $i=i^{s+1}$.\footnote{This is because $\overline{\mathrm{d}}^{k-1}_i(\mathbf{p})$ contains all agent $i$'s left-hand-side and right-hand-side primitive contracts that appear after the $(k-1)$th agent (which is not agent $i$) in complements chain, while all agent $i$'s left-hand-side and right-hand-side primitive contracts that appear before the $(k-1)$th agent in complements chain are not in $\overline{\mathrm{d}}^{k-1}_i(\mathbf{p}_i)$.}

We have $w^{k-1}\in\Phi=\overline{\mathrm{d}}^{k-1}_i(\mathbf{p}_i)$ due to Condition (i) of Definition \ref{def_chain} and $i=i^k$. If $k\leq s$, we also know $w^k\in\Phi$ since $w^k$ has not been removed before Step $k$ of the procedure for verifying whether $\mathbf{p}$ is an equilibrium price vector. The definition of $f_i(\Psi)$ indicates $w^{k-1}\notin\Psi$, and $w^{l-1}\notin\Psi$ if $l\geq2$. We know $w^l\notin\Phi$ since $w^l$ has been removed before Step $k-1$ of the procedure. If $l\geq 2$, we also know $w^{l-1}\notin\Phi$ since $w^{l-1}$ has been removed before Step $k-1$ of the procedure. Hence, we have $w^{k-1}\in \Psi\cup\Phi$, and $w^k\in \Psi\cup\Phi$ if $k\leq s$. Then by the definition of $f_i(\Psi\cup\Phi)$ we have $k\notin f_i(\Psi\cup\Phi)$. We also have $w^l\notin \Psi\cap\Phi$, and $w^{l-1}\notin \Psi\cap\Phi$ if $l\geq2$. Then by the definition of $f_i(\Psi\cap\Phi)$ we have $l\notin f_i(\Psi\cap\Phi)$. Since $f_i(\Phi)=\emptyset$, we know that $f_i(\Psi\cup\Phi)\subseteq f_i(\Psi)$\footnote{Suppose $j\in f_i(\Psi\cup\Phi)$ and $j\notin f_i(\Psi)$, then $w^{j-1}\notin\Psi\cup\Phi$ if $j\geq2$, and $w^j\in\Psi\cup\Phi$ if $j\leq s$; and thus $j\notin f_i(\Psi)$ implies $w^j\notin\Psi$ if $j\leq s$ (otherwise $w^j\in\Psi$ (if $j\leq s$) and $w^{j-1}\notin\Psi$ (if $j\geq2$) indicates $j\in f_i(\Psi)$). We thus have $w^j\in\Phi$ if $j\leq s$, and $w^{j-1}\notin\Phi$ if $j\geq2$, which indicates $j\in f_i(\Phi)$ and contradicts $f_i(\Phi)=\emptyset$.} and $f_i(\Psi\cap\Phi)\subseteq f_i(\Psi)$.\footnote{Suppose $j\in f_i(\Psi\cap\Phi)$ and $j\notin f_i(\Psi)$, then $w^{j-1}\notin\Psi\cap\Phi$ if $j\geq2$, and $w^j\in\Psi\cap\Phi$ if $j\leq s$; and thus $j\notin f_i(\Psi)$ implies $w^{j-1}\in\Psi$ if $j\geq 2$ (otherwise $w^{j-1}\notin\Psi$ (if $j\geq2$) and $w^j\in\Psi$ (if $j\leq s$) indicates $j\in f_i(\Psi)$). We thus have $w^{j-1}\notin\Phi$ if $j\geq2$, and $w^j\in\Phi$ if $j\leq s$, which indicates $j\in f_i(\Phi)$ and contradicts $f_i(\Phi)=\emptyset$.} Then, since $k\notin f_i(\Psi\cup\Phi)$, $l\notin f_i(\Psi\cap\Phi)$, and $l,k\in f_i(\Psi)$, we have $|f_i(\Psi\cup \Phi)|,|f_i(\Psi\cap \Phi)|\leq|f_i(\Psi)|-1$.
\end{proof}

We assume inductively that $\mathrm{U}_i(\Psi,\mathbf{p}'_i)\leq \mathrm{V}_i(\mathbf{p}_i)$ for any agent $i\in\{i^1,\ldots,i^s\}\setminus\{i^{s+1}\}$ and any $\Psi\in\Omega_i$ if $|f_i(\Psi)|\leq k$ where $k\geq1$. Notice that this assumption holds when $k=1$. Now, consider $\Psi\subseteq\Omega_i$ with $|f_i(\Psi)|=k+1$. By Lemma \ref{lma_phi}, there exists $\Phi\in \mathrm{D}_i(\mathbf{p}_i)$ such that $\mathrm{U}_i(\Phi,\mathbf{p}'_i)=\mathrm{V}_i(\mathbf{p}_i)$ and  $|f_i(\Psi\cup \Phi)|,|f_i(\Psi\cap \Phi)|\leq k$. Gross complementarity indicates $\mathrm{v}_i(\Psi)+\mathrm{v}_i(\Phi)\leq \mathrm{v}_i(\Psi\cup \Phi)+\mathrm{v}_i(\Psi\cap \Phi)$, which further implies $\mathrm{U}_i(\Psi,\mathbf{p}'_i)+\mathrm{U}_i(\Phi,\mathbf{p}'_i)\leq \mathrm{U}_i(\Psi\cup \Phi,\mathbf{p}'_i)+\mathrm{U}_i(\Psi\cap \Phi,\mathbf{p}'_i)$. Since $\mathrm{U}_i(\Phi,\mathbf{p}'_i)=\mathrm{V}_i(\mathbf{p}_i)$, and the inductive assumption indicates $\mathrm{U}_i(\Psi\cup \Phi,\mathbf{p}'_i),\mathrm{U}_i(\Psi\cap \Phi,\mathbf{p}'_i)\leq \mathrm{V}_i(\mathbf{p}_i)$ , we have $\mathrm{U}_i(\Psi,\mathbf{p}'_i)\leq \mathrm{V}_i(\mathbf{p}_i)$. The other part for the last agent $i^{s+1}$ of the complements chain follows the same argument.
\end{proof}

\subsection{Proofs for Section \ref{Sec_verify} }\label{Sec_proofverify}

We present the procedure for verifying equilibrium price vectors in Section \ref{Sec_verify} following the discussions in Section \ref{Sec_strong} and Section \ref{Sec_partial}. Now we provide a rigorous proof for this procedure. The following Lemma \ref{lma_large} and Lemma \ref{lma_non} show that the procedure correctly verifies whether a balanced price vector $\mathbf{p}\in \mathbb{B}$ is an equilibrium price vector and produces a stable outcome when $\mathbf{p}$ is an equilibrium price vector.

\begin{lemma}\label{lma_inter}
\normalfont
Suppose primitive contracts are gross complements for all agents, and the procedure for verifying whether $\mathbf{p}\in\mathbb{B}$ is an equilibrium price vector terminates at Step $s\geq1$. Then, if $(\Phi,\mathbf{p})$ is a competitive equilibrium, we have $\mathrm{G}^k(\mathbf{p})\cap\Phi=\emptyset$ for all $1\leq k\leq s$.
\end{lemma}

\begin{proof}
If $w\in G^1(\mathbf{p})$ is level-1 partially agreeable, we have $w\notin\overline{\mathrm{d}}^1_i(\mathbf{p}_i)$ for some $i\in N(w)$. Notice that $\overline{\mathrm{d}}^1_i(\mathbf{p}_i)$ is agent $i$'s largest demand set, and $\Phi_i\in \mathrm{D}_i(\mathbf{p}_i)$ since $(\Phi,\mathbf{p})$ is a competitive equilibrium, we thus have $\Phi_i\subseteq\overline{\mathrm{d}}^1_i(\mathbf{p}_i)$. Hence, $\Phi\cap G^1(\mathbf{p})=\emptyset$ and $\Phi\subseteq\Omega^2$. 

We are done if $s=1$. Now suppose $s\geq2$. We assume inductively $\Phi\cap G^l(\mathbf{p})=\emptyset$ for all $1\leq l\leq k$ where $k\in\{1,2,\cdots,s-1\}$. We thus have $\Phi\subseteq\Omega^{k+1}(\mathbf{p})$. If primitive contracts are gross complements for all agents, each agent's smallest demand set at $\mathbf{p}$ belongs to $\Omega^{k+1}(\mathbf{p})$ since $\mathrm{G}^l(\mathbf{p})\cap \mathrm{H}(\mathbf{p})=\emptyset$ for all $l\in\{1,2,\cdots,k\}$. Hence, $\overline{\mathrm{d}}^{k+1}_i(\mathbf{p}_i)$ is well-defined in (\ref{def_proce}) since there is at least one $\Psi^l=\underline{d}_i(\mathbf{p}_i)\in \mathrm{D}_i(\mathbf{p}_i)$ satisfying $\Psi^l\subseteq\Omega^{k+1}(\mathbf{p})$, and thus $\overline{\mathrm{d}}^{k+1}_i(\mathbf{p}_i)$ is agent $i$'s largest demand set confined in $\Omega^{k+1}(\mathbf{p})$. Then, for each agent $i\in I$, $\Phi_i\in \mathrm{D}_i(\mathbf{p}_i)$ and $\Phi\subseteq\Omega^{k+1}(\mathbf{p})$ implies $\Phi_i\subseteq\overline{\mathrm{d}}^{k+1}_i(\mathbf{p}_i)$.  If $w\in G^{k+1}(\mathbf{p})$, we have $w\notin\overline{\mathrm{d}}^{k+1}_i(\mathbf{p}_i)$ for some $i\in N(w)$. we thus know $\Phi\cap G^{k+1}(\mathbf{p})=\emptyset$. This completes the proof since the inductive assumption holds for $k=1$.
\end{proof}

\begin{lemma}\label{lma_large}
\normalfont
If primitive contracts are gross complements for all agents, and the procedure for verifying whether $\mathbf{p}\in\mathbb{B}$ is an equilibrium price vector terminates at Step $s\geq1$ as $\mathrm{G}^s(\mathbf{p})=\emptyset$, then $(\bigcup_{i\in I}\overline{\mathrm{d}}^s_i(\mathbf{p}_i),\mathbf{p})$ is a competitive equilibrium in which $\bigcup_{i\in I}\overline{\mathrm{d}}^s_i(\mathbf{p}_i)$ is the largest efficient set of primitive contracts.
\end{lemma}

\begin{proof}
We have proved this statement for $s=1$ in Section \ref{Sec_partial}. Now suppose $s\geq2$. Since $\mathrm{G}^k(\mathbf{p})\cap \mathrm{H}(\mathbf{p})=\emptyset$ for each $1\leq k\leq s$, we know that each agent $i$'s smallest demand set at $\mathbf{p}$ belongs to $\Omega^k(\mathbf{p})$: $\underline{d}_i(\mathbf{p}_i)\subseteq\Omega^k(\mathbf{p})$ for each $1\leq k\leq s$. Hence $\overline{\mathrm{d}}^k_i(\mathbf{p}_i)$ is well-defined in (\ref{def_proce}) since there is at least one $\Psi^l=\underline{d}_i(\mathbf{p}_i)\in \mathrm{D}_i(\mathbf{p}_i)$ satisfying $\Psi^l\subseteq\Omega^k(\mathbf{p})$ for all $1\leq k\leq s$. Then gross complementarity indicates $\overline{\mathrm{d}}^s_i(\mathbf{p}_i)\in\mathrm{D}_i(\mathbf{p}_i)$ for each agent $i\in I$. Since $\mathrm{G}^s(\mathbf{p})=\emptyset$ indicates $\overline{\mathrm{d}}^s_i(\mathbf{p}_i)=\Omega_i\cap(\bigcup_{i\in I}\overline{\mathrm{d}}^s_i(\mathbf{p}_i))\in\mathrm{D}_i(\mathbf{p}_i)$, we know that $(\bigcup_{i\in I}\overline{\mathrm{d}}^s_i(\mathbf{p}_i),\mathbf{p})$ is a competitive equilibrium.

The part (i) of Lemma \ref{lma_effi} indicates that the set $\bigcup_{i\in I}\overline{\mathrm{d}}^s_i(\mathbf{p}_i)$ is efficient. Suppose $\bigcup_{i\in I}\overline{\mathrm{d}}^s_i(\mathbf{p}_i)$ is not the largest efficient set of primitive contracts, then the part (i) of Lemma \ref{lma_RY} implies that there is an efficient set of primitive contracts $\Phi\supset\bigcup_{i\in I}\overline{\mathrm{d}}^s_i(\mathbf{p}_i)$. The part (ii) of Lemma \ref{lma_effi} indicates that $(\Phi,\mathbf{p})$ is also a competitive equilibrium. We know that $\Phi\subseteq\Omega^s(\mathbf{p})$ does not hold since $\overline{\mathrm{d}}^s_i(\mathbf{p}_i)$ is each agent $i$'s largest demand set confined in $\Omega^s(\mathbf{p})$. Thus, we have $\Phi\cap G^k(\mathbf{p})\neq\emptyset$ for some $1\leq k\leq s-1$. This is impossible according to Lemma \ref{lma_inter}.

Therefore, the set $\bigcup_{i\in I}\overline{\mathrm{d}}^s_i(\mathbf{p}_i)$ is the largest efficient set of primitive contracts, and thus, the part (ii) of Lemma \ref{lma_RY} indicates that the competitive equilibrium $(\bigcup_{i\in I}\overline{\mathrm{d}}^s_i(\mathbf{p}_i),\mathbf{p})$ corresponds to a stable outcome.
\end{proof}

\begin{lemma}\label{lma_non}
\normalfont
If primitive contracts are gross complements for all agents, and the procedure for verifying whether $\mathbf{p}\in\mathbb{B}$ is an equilibrium price vector terminates at Step $s\geq1$ as $\mathrm{G}^s(\mathbf{p})\cap \mathrm{H}(\mathbf{p})\neq\emptyset$, then $\mathbf{p}$ is a non-equilibrium price vector.
\end{lemma}

\begin{proof}
If $(\Phi,\mathbf{p})$ is a competitive equilibrium, then $\mathrm{H}(\mathbf{p})\subseteq\Phi$. Then $\mathrm{G}^s(\mathbf{p})\cap \mathrm{H}(\mathbf{p})\neq\emptyset$ indicates $\mathrm{G}^s(\mathbf{p})\cap \Phi\neq\emptyset$. This contradicts Lemma \ref{lma_inter}.
\end{proof}


\bigskip

\bigskip

\begin{thebibliography}{99999999999999999999999999999999999999999}

\bibitem[Arrow and Hahn(1971)]{AH71}{\small Arrow, K.J., Hahn, F.H., 1971. General competitive analysis. San Francisco: Holden-Day.}

\bibitem[Ausubel(2004)]{A04}{\small Ausubel, L., 2004. An efficient ascending-bid auction for multiple objects. \emph{American Economic Review}, 94(5), 1452-75.}

\bibitem[Ausubel(2006)]{A06}{\small Ausubel, L., 2006. An efficient dynamic auction for heterogeneous commodities. \emph{American Economic Review}, 96, 602-629.}

\bibitem[Baldwin et al.(2023)]{BEJK23}{\small Baldwin, E., Edhan, O., Jagadeesan, R., Klemperer, P., Teytelboym, A., 2023. The equilibrium existence duality. \emph{Journal of Political Economy}, 131(6), 1440-1476.}
  
\bibitem[Baldwin and Klemperer(2019)]{BK19}{\small Baldwin, E., and Klemperer, P., 2019. Understanding Preferences: ``Demand Types'', and the Existence of Equilibrium with Indivisibilities. \emph{Econometrica}, 87, 867-932. }
    
\bibitem[Bando and Hirai(2021)]{BH21}{\small Bando, K., Hirai, T., 2021. Stability and venture structures in multilateral matching. \emph{Journal of Economic Theory}, 196, 105292.}

\bibitem[Bikhchandani and Mamer(1997)]{BM97}{\small Bikhchandani, S., Mamer, J.W., 1997. Competitive equilibrium in an exchange economy with indivisibilities. \emph{Journal of Economic Theory}, 74(2), 385-413. }
    
\bibitem[Candogan et al.(2021)]{CEV21}{\small Candogan, O., Epitropou, M., Vohra, R.V., 2021. Competitive equilibrium and tranding networks: A network flow approach. \emph{Operations Research}, 69(1), 114-147.}
    
\bibitem[Danilov et al.(2001)]{DKM01}{\small Danilov, V., Koshevoy, G., Murota, K., 2001. Discrete convexity and equilibria in economics with indivisible goods and money. \emph{Mathematical Social Sciences}, 41, 251-273.}

\bibitem[Demange et al.(1986)]{DGS86}{\small Demange, G., Gale, D., Sotomayor, M., 1986. Multi-item auctions. \emph{Journal of Political Economy}, 94, 863-872.}
    
\bibitem[Fleiner et al.(2019)]{FJJT19}{\small Fleiner, T., Jagadeesan, R., Jankó, Z., Teytelboym, A., 2019. Trading Networks with Frictions. \emph{Econometrica}, 87(5), 1633-1661.}
    
\bibitem[Gale(1984)]{G84}{\small Gale, D., 1984. Equilibrium in a discrete economy with money. \textit{International Journal of Game Theory}, 13, 61-64.}
    
\bibitem[Gul and Stacchetti(1999)]{GS99}{\small Gul, F., Stacchetti, E., 1999. Walrasian equilibrium with gross substitutes. \textit{Journal of Economic Theory}, 87, 95-124.}
    
\bibitem[Gul and Stacchetti(2000)]{GS00}{\small Gul, F., Stacchetti, E., 2000. The english auction with differentiated commodities. \textit{Journal of Economic Theory}, 92, 66-95.}

\bibitem[Hatfield and Kominers(2015)]{HK15}{\small Hatfield, J.W., Kominers, S.D., 2015. Multilateral matching. \textit{Journal of Economic Theory}, 156, 175-206.}

\bibitem[Hatfield et al.(2013)]{HKNOW13}{\small Hatfield, J. W., Kominers, S. D.,  Nichifor, A., Ostrovsky, M., Westkamp, A., 2013. Stability and competitive equilibrium in trading networks. \textit{Journal of Political Economy}, 121(5), 966-1005.}
    
\bibitem[Hatfield et al.(2021)]{HKNOW21}{\small Hatfield, J. W., Kominers, S. D.,  Nichifor, A., Ostrovsky, M., Westkamp, A., 2021. Chain stability in trading networks. \textit{Theoretical Economics}, 16, 197-234.}
    
\bibitem[Huang(2023)]{H23}{\small Huang, C., 2023. Multilateral matching with scale economies. \textit{Working paper}, arXiv:2310.19479. }
    
\bibitem[Kaneko and Yamamoto(1986)]{KY86}{\small Kaneko, M., Yamamoto, Y., 1986. The existence and computation of competitive equilibria in markets with an indivisible commodity. \textit{Journal of Economic Theory}, 38, 118-136. }

\bibitem[Kelso and Crawford(1982)]{KC82}{\small Kelso, A. S., Crawford, V.P., 1982. Job matching, coalition formation and gross substitutes. \textit{Econometrica}, 50, 1483-1504. }
    
\bibitem[Klemperer(2004)]{K04}{\small Klemperer, P., 2004. Auctions: Theory and practice. The Toulouse Lectures in Economics. Princeton, NJ: Princeton University Press.} 
    
\bibitem[Krishna(2010)]{K10}{\small Krishna, V., 2010. Auction theory. Elsevier Science.}   

\bibitem[Koopmans and Beckmann(1957)]{KB57}{\small Koopmans, T.C., Beckmann, M., 1957. Assignment problems and the location of economic activities. \textit{Econometrica}, 25(1), 53-76.}
    
\bibitem[Kojima et al.(2020)]{KSY20}{\small Kojima, F., Sun, N., Yu, N.N., 2020. Job matching with subsidy and taxation. \textit{American Economic Review}, 110(9), 2935-2947.}
    
\bibitem[Kojima et al.(2024)]{KSY24}{\small Kojima, F., Sun, N., Yu, N.N., 2024. Job matching with subsidy and taxation. \textit{Review of Economic Studies}, 91(1), 372-402.}
    
\bibitem[McAfee et al.(2010)]{MMW10}{\small McAfee, R.P., McMillan, J., Wilkie, S., 2010. The greatest auction in history, in J. J. Siegfried (ed.), Better living through economics, Harvard University Press, Cambridge, MA, 168-184.}
    
\bibitem[Milgrom(2000)]{M00}{\small Milgrom, P., 2000. Putting auction theory to work: the simultaneous ascending auction. \textit{Journal of Political Economy}, 108(2), 245-72. }
    
\bibitem[Milgrom(2004)]{M04}{\small Milgrom, P., 2004. Putting auction theory to work. Cambridge: Cambridge University Press.}

\bibitem[Nguyen and Vohra(2024)]{NV24}{\small Nguyen, T., Vohra, R., 2024. (Near) Substitute preferences and equilibria with indivisibilities. \textit{Journal of Political Economy}, forthcoming.}
    
\bibitem[Quinzii(1984)]{Q84}{\small Quinzii, M., 1984. Core and competitive equilibria with indivisibilities. \textit{International Journal of Game Theory}, 13, 41-60.} 

\bibitem[Rostek and Yoder(2020)]{RY20}{\small Rostek, M., Yoder, N., 2020. Matching with complementary contracts. \textit{Econometrica}, 88(5), 1793-1824.}

\bibitem[Rostek and Yoder(2023)]{RY23}{\small Rostek, M., Yoder, N., 2023. Complementarity in matching markets and exchange economies. \textit{Working paper}. Available at SSRN: https://ssrn.com/abstract=3567616.}

\bibitem[Scarf(1960)]{S60}{\small Scarf, H., 1960. Some examples of global instability of the competitive equilibrium. \textit{International Economic Review}, 1, 157-172.}

\bibitem[Schlegel(2022)]{S22}{\small Schlegel, J.C., 2022. The Structure of Equilibria in Trading Networks with Frictions. \textit{Theoretical Economics}, 17(2), 801-839.}
    
\bibitem[Shapley and Shubik(1971)]{SS71}{\small Shapley, L.S., Shubik, M., 1971. The assignment game I: The core. \textit{International Journal of Game Theory}, 1(1), 111-130.}

\bibitem[Sun and Yang(2006)]{SY06}{\small Sun, N., Yang, Z., 2006. Equilibria and indivisibilities: Gross substitutes and complements. \textit{Econometrica}, 74, 1385-1402.}
    
\bibitem[Sun and Yang(2009)]{SY09}{\small Sun, N., Yang, Z., 2009. A double-track adjustment process for discrete markets with substitutes and complements. \textit{Econometrica}, 77, 933-952.}
    
\bibitem[Varian(1981)]{V81}{\small Varian, H.R., 1981. Dynamic Systems With Applications to Economics. in Handbook of Mathematical Economics, 1, ed. by K. J. Arrow and M. D. Intriligator. Amsterdam: NorthHolland.}

\bibitem[Yokote(2023)]{Y23}{\small Yokote, K., 2023. A critical comparison between the gross substitutes and complements conditions. \textit{Economics Letters}, 226, 111106.}
\end{thebibliography}
\end{document}